\documentclass[a4paper,10pt]{article}
\usepackage[utf8]{inputenc}
\usepackage{amsmath,amsfonts,amssymb}

\newtheorem{theorem}{Theorem}
\newtheorem{lemma}[theorem]{Lemma}
\newtheorem{definition}[theorem]{Definition}

\newtheorem{example}[theorem]{Example}

\newtheorem{proposition}[theorem]{Proposition}
\newenvironment{proof}[1][Proof]{\noindent\textbf{#1: }}{\ \rule{0.5em}{0.5em}}

\title{Weighted paths between partitions}
\author{Giovanni Rossi\\
\footnotesize{Department of Computer Science and Engineering - DISI}\\
\footnotesize{Mura Anteo Zamboni 7, Bologna 40126, Italy, giovanni.rossi6@unibo.it}}

\begin{document}

\maketitle

\begin{abstract}
How to quantify the distance between any two partitions of a finite set is an important issue in statistical classification, whenever different clustering results
need to be compared. Developing from the traditional Hamming distance between subsets or cardinality of their symmetric difference, this work considers alternative
metric distances between partitions. With one exception, all of them obtain as minimum-weight paths in the undirected graph corresponding to the Hasse diagram of the
partition lattice. Firstly, by focusing on the atoms of the lattice, one well-known partition distance is recognized to be in fact the analog of the Hamming distance
between subsets, with weights on edges of the Hasse diagram determined through the number of atoms in the unique maximal join-decomposition of partitions. Secondly,
another partition distance known as ``variation of information'' is seen to correspond to a minimum-weight path with edge weights determined by the entropy of
partitions. These two distances are next compared in terms of their upper and lower bounds over all pairs of partitions that are complements of one another. What
emerges is that the two distances share the same minimizers and maximizers, while a much rawer behavior is observed for the partition distance which does not 
correspond to a minimum-weight path. The idea of measuring the distance between partitions by means of minimum-weight paths in the Hasse diagram is further explored
by considering alternative symmetric and order-preserving/inverting partition functions (such as the the rank, in the simplest case) for assigning weights to edges.
What matters most, in such a general setting, turns out to be whether the weighting function is supermodular or else submodular, as this makes any minimum-weight path
visit the meet or else the join of the two partitions, depending on order preserving/inverting. Finally, two appendices are devoted respectively to a definition of
Euclidean distance between fuzzy partitions and the consensus partition (combinatorial optimization) problem.

\textbf{Keywords:} partition lattice, symmetric function, Hamming distance, Hasse diagram, geodesic distance, indicator function, graph of a polytope.

\textbf{MSC numbers:} 05A18, 05C12.
\end{abstract}

\section{Introduction}
Partitions are key instruments in many applicative scenarios at the interface of computer science, artificial intelligence and engineering, including pattern recognition,
data mining and bioinformatics, while also being \textit{``of central importance in the study of symmetric functions, a class of functions that pervades mathematics in
general''} \cite[p. 39]{Knuth2005} (see also \cite[Chapter 5]{RotaWay2009}, \cite{RosasSagan2006} and \cite[Chapter 7]{Stanley2012EnuCom} on symmetric function theory).
Below, symmetric functions are employed to define metric distances between partitions, which in turn are useful when different clustering results need to be compared. In
statistical classification, partitions of a data set may indeed be referred to as ``clusterings'', although the latter term relates to a richer set of structures than the
former. The issue addressed here typically arises since a local search clustering algorithm generally provides different outputs when initialized with different candidate
solutions (or inputs). On the other hand, a chosen clustering algorithm shall allow for different parametrizations, each yielding different results for the same data.
Finally, alternative clustering algorithms commonly partition the same set in alternative ways. In all these cases, a distance measure is essential for assessing the
proximity between diverse partitions \cite{Meila2007,Mirkin2013,CentralPartition}.

The issue is attracting considerable attention since the mid 60s \cite{Lerman1981,Rand1971,Renier1965}. More recently, since measuring the distance between partitions of a
population is fundamental for sibling relationship reconstruction in bioinformatics, several contributions over the last decade adopted a combinatorial approach for studying
one specific such a distance measure, here denoted by MMD as it relies on \textit{maximum matching}
\cite{Berger-Wolf+al2007,Konovalov2006,Konovalov+al2005B,Konovalov+al2005A,Sheikh+al2010}. More precisely, MMD can be shown \cite{Almudevar+1999,Day1981,Gusfield2002} to be
computable via the assignment problem \cite{KorteVygen2002}. Also, in most recent years sibship reconstruction has been tackled by means of a further partition distance
measure \cite{Brown+2012}, obtained axiomatically from information theory \cite{Meila2007} and called \textsl{variation of information} VI.

In this work, entire families of metric distances between partitions are considered, the principle aim being to have consistency and generalizations in terms of order (i.e.
lattice) theory. In fact, the general leading idea is the same as in \cite{LeclercMondjardet+1986,Leclerc1993,LeclercMondjardetLatticialConsensus,Monjardet1981},
namely to define distances between elements that are (partially) comparable in terms of a binary order relation, although attention is not limited to posets (partially
ordered sets), distributive lattices and semi-lattices, but mainly extends to the geometric lattice of partitions. Metrics for distributive lattices are usually defined in
terms of valuations (or modular lattice functions, such as the rank or cardinality of subsets in the Boolean case, see below). Conversely, valuations of the partition lattice
are constant functions \cite{Aigner79}, and therefore useless for defining metrics. Thus, the method proposed here relies on super/submodular lattice functions (referred to
as lower/upper valuations in \cite{LeclercMondjardet+1986}).

The first goal is to reproduce the traditional Hamming distance between two subsets, given by the number of atoms of the subset lattice included in either one but not in both
(i.e. the cardinality of their symmetric difference, see \cite{Bollobas86}). Such a benchmark is extended to the geometric lattice of partitions by focusing on atoms and
join-decompositions of lattice elements \cite{Aigner79,Stern99}. While every subset admits a unique such a decomposition, involving a number of atoms equal to the cardinality
(or rank) of the subset, a generic partition admits different join-decompositions, most of which redundant. The number of atoms involved in the unique maximal
join-decomposition of a partition is here referred to as the \textit{size} of that partition, yielding a function taking positive integer values, like the rank. In fact, the
two coincide for subset lattices but differ crucially for partition lattices. Roughly speaking, replacing the rank with the size yields a (i.e. the) Hamming distance between
partitions, denoted by HD. Apart from the resulting consonance in terms of ordered structures, HD and VI share important characterizing axioms (see \cite{Meila2007}).
Computable through scalar products between Boolean vectors, without any algorithmic issue, HD has a large range, and thus fine measurement sensitivity too.

The traditional Hamming distance between two subsets of a $n$-set is also the length of a shortest path between them in the Hasse diagram of the Boolean lattice of subsets.
Such a diagram is in fact the graph of the polytope \cite{BronstedConvex,Branko2001} given by the $n$-dimensional unit hypercube $[0,1]^n$, thus it has $2^n$ vertices that
bijectively correspond to subsets, and an edge links any two vertices when the two corresponding subsets are comparable in terms of the covering relation (see
\cite{Bollobas86,AlgeGraph} and below). In order to have exactly the same for the Hamming distance between partitions, these latter must be seen to correspond bijectively to
those graphs on $n$ (labelled) vertices each of whose components is complete. More precisely, denoting by $K_N=(N,N_2)$ the complete graph on vertex set $N=\{1,\ldots ,n\}$,
with $N_2=\{\{i,j\}:1\leq i<j\leq n\}$, partitions correspond bijectively to those graphs $G=(N,E),E\subseteq N_2$ each of whose components is a maximal complete subgraph or
a clique, and the geometric lattice of partitions of $N$ is the so-called polygon matroid defined on the edges of $K_N$ \cite[pp. 54, 259]{Aigner79}. The associated Hasse
diagram is thus recognized to be the graph of a polytope strictly included in the $\binom{n}{2}$-dimensional unit hypercube $[0,1]^{\binom{n}{2}}$. Specifically, the
$2^{\binom{n}{2}}$-set $\{0,1\}^{\binom{n}{2}}$ of hypercube vertices identifies the $2^{\binom{n}{2}}$-set of distinct graphs on vertex set $N$, whereas linear dependence
\cite{Whitney1935} entails that partitions only span $\mathcal B_n<2^{\binom{n}{2}}$ hypercube vertices, where $\mathcal B_n$ is the \textit{Bell number} of partitions of a
$n$-set ($n>1$) \cite{Aigner79,GrahamKnuth+1994,Rota1964}. While the covering relation between subsets assigns a unit weight to every edge of the $n$-cube \cite{Sebo+04},
edges of the polytope of partitions must be weighted through the size, which matches precisely the number of edges of the $\binom{n}{2}$-cube that collapse into a unique
edge of the included polytope. With these weights, the Hamming distance HD between partitions (like between subsets) is the minimum weight of a path connecting them.

The analysis then continues by observing that the size may be replaced with any alternative symmetric and (strictly) order-preserving/inverting partition function, such as
rank, entropy, logical entropy \cite{EllermanLogicalEntropy,EllermanLogic} and co-size (see below). Then, polytope edges have weights obtained as the difference between the
greater and the smaller value taken by the chosen function on the associated endpoints. Accordingly, the distance between two partitions remains the minimum weight of a path
connecting them. In particular, if the function assigning weights to edges is order-preserving and supermodular (like the size) or else submodular (like the rank), then the
minimum-weight path between any two partitions visits their meet or else their join, respectively. Analog results obtain for order-inverting and symmetric functions which
are either supermodular or else submodular.

Section 2 outlines the needed background, with emphasis on lattice functions and Hamming distances in general, while Section 3 introduces the proposed Hamming distance between
partitions, including an axiomatic characterization. Section 4 is devoted to bounding both the Hamming and variation-of-information distances over all pairs of partitions that
are complements of one another. Section 5 frames distances as minimum-weight paths in the Hasse diagram. Section 6 considers two further functions assigning weights to edges,
namely logical entropy and co-size. Sections 7 and 8 are two appendices detailing respectively a definition of Euclidean distance between fuzzy partitions and an exact
solution for the consensus partition (combinatorial optimization) problem. Section 9 concludes the paper with some final remarks.

\section{Preliminaries}
Throughout this work, the general concern is with metric distances $d(x,y)$ between elements $x,y\in X$ of a poset $(X,\geqslant)$, i.e. a (finite) set $X$ endowed with a
partial order relation $\geqslant$. Additionally, $X$ shall also be endowed with the meet $\wedge$ and join $\vee$ operators, so that $(X,\wedge,\vee)$ is a (complete) lattice
(see \cite{DaveyPriestley}). The ordered structures to be considered are grounded on a finite set $N=\{1,\ldots ,n\}$, where integers $1,\ldots ,n$ possibly denote the indices
of a data set. In particular, attention is going to be placed on the Boolean lattice $(2^N,\cap,\cup)$ of subsets of $N$ ordered by inclusion $\supseteq$ and, mostly, on the
geometric lattice $(\mathcal P^N,\wedge,\vee)$ of partitions of $N$ ordered by coarsening $\geqslant$ (see \cite{Aigner79,Stern99}). Generic subsets and partitions are denoted
respectively by $A,B\in 2^N$ and $P,Q\in\mathcal P^N$. Recall that a partition $P=\{A_1,\ldots ,A_{|P|}\}$ is a collection of (non-empty) pair-wise disjoint subsets, called
blocks, whose union is $N$. For any $P,Q\in\mathcal P^N$, if $P\geqslant Q$, then every block $B\in Q$ is included in some block $A\in P$, i.e. $A\supseteq B$. Hence the bottom
partition is $P_{\bot}=\{\{1\},\ldots ,\{n\}\}$ (like the bottom subset is $\emptyset$), while the top one is $P^{\top}=\{N\}$ (like $N$ is the top subset). Also, among
partitions the meet $\wedge$ is the coarsest-finer-than operator, while the join $\vee$ is the finest-coarser-than operator. The number $|\mathcal P^N|=\mathcal B_n$ of
partitions of $N$ is defined recursively by $\mathcal B_0:=1$ and $\mathcal B_n=\sum_{0\leq k<n}\binom{n-1}{k}\mathcal B_k$ (see \cite{Aigner79,GrahamKnuth+1994,Rota1964} on
Bell numbers).

For all ordered pairs $(x,y)\in X\times X$ of poset elements, the associated interval or segment is $[x,y]=\{z:x\leqslant z\leqslant y\}\subset X$, and $y$ is said to cover $x$,
denoted by $y>^*x$, if $[x,y]=\{x,y\}$. The Hasse diagram of poset $(X,\geqslant)$ is the graph $G=(X,E)$ whose vertices are elements $x\in X$ and edges are given by the covering
relation, i.e. $E=\{\{x,y\}:[x,y]=\{x,y\}\}$. Although these edges are sometimes assumed to be directed, thereby also indicating what elemets are covered/covering, still in the
present setting they are more fuitfully regarded as undirected, for this allows to consider paths where edges may be used in both directions. In fact, the distance between any
two vertices in a graph is the length of any shorthest path between them. More generally, if the graph is weighted, meaning that every edge has an associated (strictly positive)
weight, then the distance between any two vertices is the weight of a lightest path between them, where the weight of a path is the sum over its edges of their weight.

In a lattice $(X,\wedge,\vee)$ with bottom element $x_{\bot}$, the set $X_{\mathcal A}=\{x:x>^*x_{\bot}\}$ of atoms consists of all lattice elements that cover the bottom one.
In atomic lattices, every element $x\in X$ admits a decomposition $x=a_1\vee\cdots\vee a_k$ as a join of atoms $a_1,\ldots ,a_k\in X_{\mathcal A}$. Both the Boolean lattice
$(2^N,\cap,\cup)$ of subsets of $N$ and the geometric lattice $(\mathcal P^N,\wedge,\vee)$ of partitions of $N$ are atomic. For the former, atoms are the $n$ singletons
$\{i\},i\in N$. For the latter, atoms are the $\binom{n}{2}$ partitions consisting of $n-1$ blocks, out of which $n-2$ are singletons while the remaining one is a pair. Most
importantly, every subset $A\in 2^N$ admits a unique join-decomposition, namely $A=\cup_{i\in A}\{i\}$. Conversely, partitions generally admit several join-decompositions.
However, every partition $P\in\mathcal P^N$ admits a unique maximal join-decomposition, which includes all atoms finer than $P$. In the sequel, a great deal of attention shall
be placed on such a number of atoms finer than any given partition, to be referred to as the size of partitions.

\subsection{Lattice functions}
In order to consider alternative weights over the edges of the Hasse diagram, it is necessary to deal with different lattice functions $f:X\rightarrow\mathbb R_+$. Firstly,
from a geometric perspective, $f\in\mathbb R^{|X|}_+$ is a point in a vector space. A well-known basis of this vector space is $\{\zeta_x:x\in X\}$, where
$\zeta_x(y)=\left\{\begin{array}{c}1\text{ if }x\leqslant y\\0\text{ if }x\not\leqslant y \end{array}\right.$ for all $y\in X$. Thus, any $f$ is a linear combination
$f=\sum_{x\in X}\zeta_x\mu^f(x)$ of basis elements, with coefficients $\mu^f(x),x\in X$ given by M\"obius inversion $\mu^f:X\rightarrow\mathbb R$, where this latter obeys the
following recursion: $\mu^f(x)=f(x)-\sum_{y<x}\mu^f(y)$ for all $x\in X$, hence $\mu^f(x_{\bot})=f(x_{\bot})$ and $f(x)=\sum_{y\leqslant x}\mu^f(y)$ for all $x\in X$ (see
\cite{Aigner79,RotaMobius,Stern99}).

A lattice function $f$ is said to be:
\begin{itemize}
\item strictly order-preserving if $f(x)>f(y)$ for all $x,y\in X$ such that $x>y$,
\item strictly order-inverting if $f(x)>f(y)$ for all $x,y\in X$ such that $x<y$,
\item supermodular if $f(x\vee y)+f(x\wedge y)-f(x)-f(y)\geq 0$ for all $x,y\in X$,
\item submodular if $f(x\vee y)+f(x\wedge y)-f(x)-f(y)\leq 0$ for all $x,y\in X$,
\item modular if $f(x\vee y)+f(x\wedge y)-f(x)-f(y)=0$ for all $x,y\in X$,
\item totally positive if $\mu^f(x)\geq 0$ for all $x\in X$.
\end{itemize}
\textbf{Observation:} if $f$ is totally positive, then it is supermodular. To see this, firstly note that if $x$ and $y$ are comparable, i.e. say $x\geqslant y$, then there is
nothing to show as $x\wedge y=y$ and $x\vee y=y$, and thus the inequality defining supermodularity is satisfied with equality. Apart from this trivial case, if $x$ and $y$ are
uncomparable, i.e. $x\not\geqslant y\not\geqslant x$, then substituting the general M\"obius inversion formula above, i.e. $f(x)=\sum_{y\leqslant x}\mu^f(y)$, into the inequality
defining supermodularity formula yields $f(x\vee y)+f(x\wedge y)-f(x)-f(y)=$
\begin{eqnarray*}
&=&\sum_{z\leqslant x\vee y}\mu^f(z)+\sum_{z\leqslant x\wedge y}\mu^f(z)-\sum_{z\leqslant x}\mu^f(z)-\sum_{z\leqslant y}\mu^f(z)=\\
&=&\sum_{x\wedge y<z\leqslant x\vee y}\mu^f(z)-\sum_{x\wedge y<z\leqslant x}\mu^f(z)-\sum_{x\wedge y<z\leqslant y}\mu^f(z)=\\
&=&\underset{x\not\geqslant z\not\leqslant y}{\sum_{x\wedge y<z\leqslant x\vee y}}\mu^f(z)\geq 0
\end{eqnarray*}
where of course $[x\wedge y,x]\cap[x\wedge y,y]=\{x\wedge y\}$ by definition of meet.

Further lattice functions to be considered are symmetric ones, i.e. those that are invariant under the action of the symmetric group $\mathcal S(N)$ consisting of all $n!$
permutations $\pi:N\rightarrow N$ (see \cite[p. 161]{Aigner79}). Symmetric functions are generally very important in mathematics; for reasons of space only essential facts are here
exposed, with focus on lattices $(2^N,\cap,\cup)$ and $(\mathcal P^N,\wedge,\vee)$. For any $A\in 2^N,\pi\in\mathcal S(N)$, let $\pi A=\{\pi^{-1}(i):i\in A\}$, where $j=\pi^{-1}(i)$
is the index mapped into the $i$-th position by $\pi$. A set function $v:2^N\rightarrow\mathbb R_+$ is symmetric if $v(A)=v(\pi A)$ for all $A\in 2^N,\pi\in\mathcal S(N)$. Thus, $v$
is symmetric if $v(A)=v(B)$ for all $A,B\in 2^N$ such that $|A|=|B|$. As for partitions, for every $P\in\mathcal P^N$ let $c^P=(c^P_1,\ldots,c^P_n)\in\mathbb Z_+$ be the class or
type of $P$ (see \cite{RotaMobius}), that is to say $c^P_k=|\{B:B\in P,|B|=k\}|,1\leq k\leq n$. For all $\{B_1,\ldots ,B_{|P|}\}=P\in\mathcal P^N$ and $\pi\in\mathcal S(N)$, let
$\pi P=\{\pi B_1,\ldots \pi B_{|P|}\}$. A partition function $h:\mathcal P^N\rightarrow\mathbb R_+$ is symmetric if $h(P)=h(\pi P)$ for all $P\in\mathcal P^N,\pi\in\mathcal S(N)$.
Thus, $h$ is symmetric if $h(P)=h(Q)$ for all $P,Q\in\mathcal P^N$ such that $c^P=c^Q$.

\subsection{Hamming distance between subsets}
First of all recall that measures of the distance between elements of any (i.e. possibly non-ordered) set are referred to as ``Hamming distances'' when these elements are represented
as arrays or matrices and the distance between two of them is the number of entries where their array or matrix representations differ. The issue introduced in Section 1, namely how to
measure a distance $d(P,Q)$ between any two partitions $P,Q\in\mathcal P^N$, is firstly addressed in the following Section 3 by reproducing the traditional Hamming distance $|A\Delta B|$
between subsets $A,B\in 2^N$, where $|A\Delta B|=|A\cup B|-|A\cap A|$. This distance measure can also be expressed as $|A\Delta B|=|A\backslash B|+|B\backslash A|=r(A\cup B)-r(A\cap B)$,
where $r:2^N\rightarrow\mathbb Z_+$ is the rank function, i.e. $r(A)=|A|$ for all $A\in2^N$. The essential combinatorial feature of $|A\Delta B|$ is that it counts how many atoms
$\{i\},i\in N$ of Boolean lattice $(2^N,\cap ,\cup)$ are included in either $A$ or else $B$ but not in both. Also, $|A\Delta B|$ is a Hamming distance since subsets $A,B\in 2^N$
are represented as Boolean $n$-vectors $\chi_A,\chi_B\in\{0,1\}^n$, with characteristic function $\chi_A:N\rightarrow\{0,1\}$ defined by $\chi_A(i)=1$ if $i\in A$ and $\chi_A(i)=0$ if
$i\in N\backslash A=A^c$, for all $A\in 2^N$. Thus, $|A\Delta B|=\sum_{i\in N}\left(\chi_A(i)-\chi_B(i)\right)^2$ is precisely the number of entries where $\chi_A$ and $\chi_B$ differ
\cite{Aigner79,Bollobas86}. Evidently, characteristic functions $\chi_A,A\in 2^N$ provide a bijection between the $2^n$-set of subsets $A\in 2^N$ and the vertices $\chi_A\in\{0,1\}^n$ of
the $n$-dimensional unit hypercube $[0,1]^n$. In fact, the graph of this latter polytope \cite{BronstedConvex,Branko2001} is the Hasse diagram of Boolean lattice $(2^N,\cap,\cup)$, the
two sharing the same vertices and edges, and $|A\Delta B|$ is the length of a shortest path connecting vertices $\chi_A$ and $\chi_B$. Clearly, a shortest path is also a minimum-weight
path as long as each edge has unit weight, which is precisely what happens when edges are weighted by the rank.

For any two points $p,q\in[0,1]^n$ in the unit $n$-cube, let $\langle p,q\rangle=\sum_{1\leq i\leq n}p_iq_i$ denote their scalar product. Since $\chi_N\in\{0,1\}^n$ is the $n$-vector all
of whose entries equal 1, for all $A\in 2^N$ it holds $r(A)=|A|=\langle\chi_A,\chi_N\rangle$. Three further expressions for the Hamming distance between subsets $A,B\in 2^N$ are
$|A\Delta B|=$
\begin{eqnarray}
=|A|+|B|-2|A\cap B|&=&\langle\chi_A,\chi_N\rangle+\langle\chi_B,\chi_N\rangle-2\langle\chi_A,\chi_B\rangle=\\
&=&\langle\chi_A,\chi_N\rangle+\langle\chi_B,\chi_N\rangle-2\langle\chi_{A\cap B},\chi_N\rangle=
\end{eqnarray}
\begin{equation*}
=2|A\cup B|-|A|-|B|=2[\langle\chi_A,\chi_N\rangle+\langle\chi_B,\chi_N\rangle-\langle\chi_A,\chi_B\rangle]-\langle\chi_A,\chi_N\rangle-\langle\chi_B,\chi_N\rangle\text .
\end{equation*}
Furthermore, the following two observations are immediately checked.
\begin{itemize}
\item $r:2^N\rightarrow\{0,1,2,\ldots,n\}$ is a strictly order-preserving, symmetric and modular lattice (i.e. set) function, and
\item $|\cdot\Delta\cdot|:2^N\times 2^N\rightarrow\{0,1,2,\ldots,n\}$ is a metric: for all $A,A',B\in 2^N$,
\begin{enumerate}
\item $|A\Delta B|=|B\Delta A|$,
\item $|A\Delta B|\geq 0$, with equality if and only if $A=B$,
\item $|A\Delta A'|+|A'\Delta B|\geq|A\Delta B|$, or triangle inequality.
\end{enumerate}
\end{itemize}

\section{Partition distances}
In a (simple) graph $G=(V,E)$ with vertex set $V=\{v_1,\ldots ,v_m\}$ the edge set $E\subseteq V_2=\{\{v_i,v_j\}:1\leq i<j\leq m\}$ is included in the $\binom{m}{2}$-set of unordered pairs
of vertices. As already mentioned, the complete graph on these $m$ (labelled) vertices is $K_m=(V,V_2)$, and the Hamming distance $HD(P,Q)$ between partitions defined in the sequel
reproduces $|A\Delta B|$ while keeping into account that partitions of $N$ correspond bijectively to those graphs with vertex set $V=N$ whose components are each a complete subgraph
\cite{Aigner79}.

The combinatorial analog of $|A\Delta B|$ in terms of partitions $P,Q$, namely the number of atoms of $(\mathcal P^N,\wedge,\vee)$ finer than either $P$ or $Q$ but not finer than both,
exists in the literature \cite{Meila2007,Rand1971}, but is commonly not recognised to be such an analog. Conversely, the name ``Hamming distance between partitions'' is often customarily
maintained for a metric obtained by representing partitions $P$ as Boolean matrices $M^P\in\{0,1\}^{n\times n}$, despite these latter correspond in fact to generic binary relations on $N$
\cite[p. 393]{Mirkin1996}. Since partitions only correspond to equivalence relations, it is readily seen there are $2^{n^2}-\mathcal B_n$ binary relations which are \textit{not} equivalence
relations, yielding both conceptual and quantitative ambiguities (detailed below). In addition, the metric obtained by representing partitions $P$ as matrices $M^P$ does not yield any
shortest path between vertices of the Hasse diagram of partitions. In general, it seems desirable that the distance between elements of a ordered set (such as $2^N$ and $\mathcal P^N$) is
measured in terms of the order relation, like $|A\Delta B|$ is specified in terms of $\supseteq$. That is to say, in formal notation,
$|A\Delta B|=|\{\{i\}:A\supseteq\{i\}\not\subseteq B\}|+|\{\{i\}:A\not\supseteq\{i\}\subseteq B\}|$.

There exist many partition distance measures available in the literature, \cite[Sections 10.2, 10.3, pp. 191-193]{DezaDeza2013}, \cite[Chapter 5]{Mirkin1996}
\cite{Day1981,HubertArabie1985,Warrens2008}. Towards a clear disambiguation between the so-called \textit{Hamming distance between (matrices representing) partitions}
\cite{Meila2007,Mirkin+1970,Mirkin+2008} mentioned above and what is proposed here, recall that a binary relation $\mathcal R$ on $N$ is a subset $\mathcal R\subseteq N\times N$ of
\textit{ordered} pairs $(i,j)$ of elements $i,j\in N$ (hence unordered pairs $\{i,j\}$ satisfy $\{i,j\}=\{j,i\}$, while $(i,j)\neq(j,i)$ for ordered ones). The collection of all
such binary relations is a Boolean lattice $(2^{N\times N},\cap ,\cup)$. If symmetry $(i,j)\in\mathcal R\Rightarrow (j,i)\in\mathcal R$ and  \textit{transitivity}
$(i,j),(j,i')\in\mathcal R\Rightarrow (i,i')\in\mathcal R$ hold, then $\mathcal R$ is an \textit{equivalence} relation, or a partition of $N$ into equivalence classes: $\supseteq$-maximal
subsets $A\in 2^N$ such that $(i,j),(j,i)\in\mathcal R$ for all $i,j\in A$ are precisely its blocks. A binary relation $\mathcal R$ may be represented as a Boolean matrix
$M^{\mathcal R}\in\{0,1\}^{n\times n}$ with entries $M^{\mathcal R}_{ij}=1$ if $(i,j)\in\mathcal R$ and $M^{\mathcal R}_{ij}=0$ if $(i,j)\not\in\mathcal R$. Now let two equivalence
relations $\mathcal R^P,\mathcal R^Q$ have associated partitions $P,Q$ and representing matrices $M^{\mathcal R^P},M{^{\mathcal R^Q}}$. The distance $d(\mathcal R^P,\mathcal R^Q)$ between
subsets $\mathcal R^P,\mathcal R^Q\in 2^{N\times N}$ can be computed as
$d(\mathcal R^P,\mathcal R^Q)=|\mathcal R^P\Delta\mathcal R^Q|=|\mathcal R^P\cup\mathcal R^Q|-|\mathcal R^P\cap\mathcal R^Q|$. This is the number of 1s in matrix
$M^{\mathcal R^P\Delta\mathcal R^Q}=M^{\mathcal R^P}+M^{\mathcal R^Q}$ modulo 2. While providing a distance between partitions $P$ and $Q$, this is in fact the traditional Hamming
distance between certain subsets $R^P,R^Q\in 2^{N\times N}$, while generic such subsets $\mathcal R\in2^{N\times N}$ correspond to partitions only in very special cases, as
lattice $(2^{N\times N},\cap ,\cup)$ contains $2^{n^2}-\mathcal B_n$ elements, or binary relations, that do not correspond to partitions, or equivalence relations. The argument
also applies when partitions are represented as Boolean $n\times n$-matrices through the complement $\bar{\mathcal R}$ of equivalence relations $\mathcal R$, known as
\textit{apartness relations} in computer science \cite{EllermanLogicalEntropy,EllermanLogic}, i.e. $\bar{\mathcal R}^P=(N\times N)\backslash\mathcal R^P$ (this is detailed below).

The point is that in finite sets such as $2^N,\mathcal P^N$ and $2^{N\times N}$ where there is no ``natural'' metric (like the Euclidean norm in $\mathbb R^m$), the distance between elements
$x$ and $y$ must be quantified, in some way, by the number of elements $z$ between $x$ and $z$, where ``between'' means that $z$ must be comparable, in terms of the order relation, with
$x$ and/or $y$. To achieve this, in the present setting, consider that the partition lattice $(\mathcal P^N,\wedge,\vee)$ is a matroid (see \cite{Aigner79,Stern99} and above). However
regarded, it is necessarily embedded into a larger subset lattice, with which some elements are shared while some others are not. Apart from binary relations just described, a na\"ive
example comes from noticing that partitions $P$ are collections of subsets, i.e. $P\in 2^{2^N}$, and thus the distance between $P$ and $Q$ might be computed as the Hamming distance
$|P\Delta Q|$ between elements of subset lattice $(2^{2^N},\cap ,\cup)$, i.e. the number of subsets $A\in 2^N$ that are blocks of either one but not both. Again, there are really many (i.e.
$2^{2^n}-\mathcal B_n$) set systems (or collections $\mathcal S\in 2^{2^N}$ of subsets) that do not correspond to partitions. This feature is maintained even when $P$ and $Q$ are decomposed
as joins of atoms, for they generally admit several such join-decompositions \cite[Chapter II]{Aigner79}. Yet, when regarded from this perspective partition lattice
$(\mathcal P^N,\wedge ,\vee)$ is seen to be included in subset lattice $(2^{N_2},\cap ,\cup)$, with the two sharing the same $\binom{n}{2}$ atoms. In fact, $2^{N_2}$ is the \textit{minimal}
Boolean lattice including the partition lattice. Accordingly, the Hamming distance between partitions HD proposed below relies precisely on representing partitions as Boolean
$\binom{n}{2}$-vectors, although only $\mathcal B_n<2^{\binom{n}{2}}$ distinct such vectors correspond to partitions. In particular, HD is the traditional Hamming distance $|E\Delta E'|$
between edge sets $E,E'\in 2^{N_2}$ of graphs on vertex set $N$, with these latter corresponding to partitions only when in both graphs $G=(N,E),G'=(N,E')$ each component is a complete
subgraph.

\subsection{Hamming distance between partitions}
In combinatorial theory, both $(2^N,\cap,\cup)$ and $(\mathcal P^N,\wedge,\vee)$ are geometric lattices \cite[p. 54]{Aigner79}. As such, they are atomic, meaning that every
element is decomposable as a join of atoms (see above). The rank function $r:\mathcal P^N\rightarrow\mathbb Z_+$ of the partition lattice is $r(P)=n-|P|$, with height
$r(P^{\top})=n-1$ and $r(P_{\bot})=0$ for the top and bottom elements, respectively. As already outlined, atoms are immediately above $P_{\bot}$, with rank $1$, in the associated
Hasse diagram \cite[p. 889]{Meila2007}, where coarser partitions occupy upper levels. Thus, atoms are those partitions consisting of $n-1$ blocks, namely $n-2$ singletons and one
pair. These $\binom{n}{2}$ pairs $\{i,j\}\in N_2$ are the same atoms as in Boolean lattice $(2^{N_2},\cap ,\cup)$. Notationally, it is now convenient to let $[ij]\in\mathcal P^N$
be the atom where the unique $2$-cardinal block is pair $\{i,j\}\in[ij]$ (this is denoted by $\pi_{xy}$ in \cite[p. 150]{LeclercMondjardetLatticialConsensus}, where $x,y$ are elements of
the partitioned set while $\pi$ denotes the generic partition).

In order to have a combinatorially congruhent reproduction of the Hamming distance between partitions, let $\mathcal P^N_{\mathcal A}=\{[ij]:1\leq i<j\leq n\}$ be the $\binom{n}{2}$-set of atoms
of the partition lattice, with isomorphism $\mathcal P^N_{\mathcal A}\cong N_2$. The analog of characteristic function $\chi_A$ is \textit{indicator function}
$I_P:\mathcal P^N_{\mathcal A}\rightarrow\{0,1\}$, defined by
\begin{equation*}
I_P([ij])=\left\{\begin{array}{c}1\text{ if }P\geqslant [ij]\\
0\text{ if }P\not\geqslant [ij] \end{array}\right .
\text{ for all }P\in\mathcal P^N,[ij]\in\mathcal P^N_{\mathcal A}\text .
\end{equation*}  
In words, if pair $\{i,j\}$ is included in some block $A$ of $P$, i.e. $\{i,j\}\subseteq A\in P$, then partition $P$ is coarser than atom $[ij]$, and the corresponding position
$I_P([ij])$ of indicator array $I_P$ has entry $1$. Otherwise, that position is $0$. For the top partition $P^{\top}=\{N\}$, indicator function $I_{P^{\top}}$ is the
$\binom{n}{2}$-vector with all entries equal to 1. For the bottom partition $P_{\bot}$, analogously $I_{P_{\bot}}\in\{0,1\}^{\binom{n}{2}}$ is the $\binom{n}{2}$-vector all of whose
entries equal 0. The number $s(P)=|\{[ij]:[ij]\leqslant P\}|$ of atoms finer than any partition $P$ is \cite{Rossi2011} the \textit{size} $s:\mathcal P^N\rightarrow\mathbb Z_+$
mentioned in Section 1, i.e.
\begin{equation*}
s(P)=\sum_{A\in P}\binom{|A|}{2}=\sum_{1\leq k\leq n}c^P_k\binom{k}{2}=\langle I_P,I_{P^{\top}}\rangle\text .
\end{equation*}
While the cardinality $|A|=\langle\chi_A,\chi_N\rangle$ of subsets takes every integer value between $0$ and $n$, the size $s(P)=\langle I_P,I_{P^{\top}}\rangle$ of partitions
does not the same between $0$ and $\binom{n}{2}$. Minimally, this is already observable for $N=\{1,2,3\}$, as there are $\mathcal B_3=5$ partitions: the finest
$\{\{1\},\{2\},\{3\}\}$ and coarsest $\{1,2,3\}$ ones, together with the $\binom{3}{2}=3$ atoms $[12]=\{\{1,2\},\{3\}\}$, $[13]=\{\{1,3\},\{2\}\}$ and $[23]=\{\{2,3\},\{1\}\}$.
Thus, there is no partition with size equal to $2$, as $[12]\vee[23]=[12]\vee[13]=[13]\vee[23]=\{1,2,3\}=[12]\vee[13]\vee[23]$. Available sizes of partitions of a $n$-set, for
$1\leq n\leq 7$, are in Table 1 below.
\begin{table}[htbp]
\caption{\textsl{Available sizes of partitions of a $n$-set, $1\leq n\leq 7$.}}
\label{tab: available sizes}
\begin{center}
\begin{tabular}{|c|c|c|}
\hline
$n$ & $\{s(P):P\in\mathcal P^N\}$ (available sizes)\\
\hline
1& $\{0\}$\\
\hline
2 & $\{0,1\}$\\
\hline
3 & $\{0,1,3\}$\\
\hline
4 & $\{0,1,2,3,6\}$\\
\hline
5 & $\{0,1,2,3,4,6,10\}$\\
\hline
6 & $\{0,1,2,3,4,6,7,10,15\}$\\
\hline
7 & $\{0,1,2,3,4,5,6,7,9,10,11,15,21\}$\\
\hline
\end{tabular}
\end{center}
\end{table}

Both lattices $\mathcal P^N$ and $2^{N_2}$ are atomic, with every element $P\in\mathcal P^N$ and $E\in 2^{N_2}$ admitting a decomposition as a join of atoms. Yet, while subsets
$E\in 2^{N_2}$ (or edge sets of graphs with vertex set $N$) admit a unique such a decomposition, namely $E=\cup_{\{i,j\}\in E}\{i,j\}$, partitions generally admit several such
decompositions $P=[ij]_1\vee\cdots\vee[ij]_k$. For $n=3$ as above, the coarsest partition $\{1,2,3\}$ decomposes either as the join of any two atoms, or else as the join of all
the three available atoms at once. In particular, the rank $r(P)$ of $P$ is the \textit{minimum} number of atoms involved in a join-decomposition of $P$, while the size $s(P)$ is
the \textit{maximum} number of atoms involved in such a decomposition. Hence, the coarsest partition $\{1,2,3\}$ of a $3$-cardinal set has rank $r(\{1,2,3\})=3-1=2$ and size
$s(\{1,2,3\})=3=\binom{3}{2}$. 

The rank $r(P)$ of partitions is well-known to be strictly order-preserving, symmetric and submodular, while the size $s(P)$ is strictly order-preserving, symmetric and
supermodular. This is shown below.

\begin{lemma}
The size is a strictly order-preserving partition function: if $P>Q$, then $s(P)>s(Q)$, for all $P,Q\in\mathcal P^N$.
\end{lemma}
\begin{proof}
If $P>Q$, then every $A\in P$ is the union of some $B_1,\ldots ,B_{k_A}\in Q$, i.e. $A=B_1\cup\cdots\cup B_{k_A}$, with $k_A>1$ for at least one $A\in P$. The
union $B\cup B'$ of any $B,B'\in Q$ increases the size by
\begin{equation*}
\binom{|B|+|B'|}{2}-\left(\binom{|B|}{2}+\binom{|B'|}{2}\right)=|B||B'|\text ,
\end{equation*}
which is strictly positive as blocks are non-empty.
\end{proof}

In order to reproduce expressions (1-2) of Section 2.2 above, Hamming distance HD between partitions has to count the number of atoms finer than either one of any two partitions
but not finer than both. Thus, in terms of cardinalities of subsets of atoms, distance $HD:\mathcal P^N\times\mathcal P^N\rightarrow\mathbb Z_+$ is given by
\begin{equation*}
HD(P,Q)=|\{[ij]:P\geqslant [ij]\not\leqslant Q\}|+|\{[ij]:P\not\geqslant [ij]\leqslant Q\}|\text .
\end{equation*}
The size and the indicator function allow to obtain HD as follows:
\begin{equation}
HD(P,Q)=s(P)+s(Q)-2s(P\wedge Q)=\langle I_P,I_{P^{\top}}\rangle+\langle I_Q,I_{P^{\top}}\rangle-2\langle I_P,I_Q\rangle\text.
\end{equation}
Also note that $P\wedge Q=\underset{P\geqslant [ij]\leqslant Q}{\vee}[ij]$, and this is the maximal decomposition of $P\wedge Q$ as a join of atoms, namely that involving
$s(P\wedge Q)$ atoms. Therefore,
\begin{equation}
HD(P,Q)=\langle I_P,I_{P^{\top}}\rangle+\langle I_Q,I_{P^{\top}}\rangle-2\langle I_{P\wedge Q},I_{P^{\top}}\rangle\text.
\end{equation}
In view of expressions (1-4), there seems to remain no doubt that, from a combinatorial perspective, $HD(P,Q)$ is in fact the faithful translation of the traditional Hamming
distance $|A\Delta B|$ from subsets $A,B$ to partitions $P,Q$.

\subsection{Two further partition distances}

Two non-Hamming partition distances are now briefly introduced, since they provide a term of comparison for the following sections and also in view of the recent literature in
bioinformatics cited in Section 1. Any subset $A$ has a unique complement $A^c=N\backslash A$. For all partitions $P$ and all non-empty subsets $A\neq\emptyset$, let
$P^A=\{B\cap A:B\in P,\emptyset\neq B\cap A\}$ denote the partition of $A$ induced by $P$. Maximum matching distance $MMD(P,Q)$ between partitions $P,Q$ is
\begin{equation}
MMD(P,Q)=\min\{|A^c|:\emptyset\subset A\subseteq N,P^A=Q^A\}\text .
\end{equation}
This is the minimum number of elements $i\in N$ that must be deleted in order for the two residual induced partitions to coincide. Also, $MMD(P,Q)$ \textit{``is the minimum number
of elements that must be moved between clusters of $P$ so that the resulting partition equals $Q$''} \cite[p. 160]{Gusfield2002}. It is computable as a maximum matching or
assignment problem \cite{Day1981}, \cite[chapter 11]{KorteVygen2002}. In a graph a matching is a set of pairwise disjoint edges, i.e. the endpoints are all different vertices. Now consider the
bipartite graph $G=(P\cup Q,E)$ with $|P|+|Q|$ vertices, one for each block of each partition, and join any two of them $A\in P$ and $B\in Q$ with an edge $\{A,B\}\in E$ if
$A\cap B\neq\emptyset$. In addition, let $|A\cap B|$ be the weight of the edge. Then, determining $MMD(P,Q)$ amounts to find a maximum-weight matching $E^*$ in $G$, that is one
where the sum $\sum_{(A,B)\in E^*}|A\cap B|$ of edge weights is maximal. In fact, the minimum number $MMD(P,Q)$ of elements that must be removed for the two residual partitions to
coincide is the sum $\sum_{(A,B)\in E^*}|A\Delta B|$ over all selected edges of the cardinality of the symmetric difference between the associated endpoints.

Another important measure of the distance between any two partitions $P$ and $Q$ is the variation of information $VI(P,Q)$, obtained axiomatically from information theory (see
\cite[Expressions (15)-(22), pages 879-80]{Meila2007}). Entropy $e(P)=-\sum_{A\in P}\frac{|A|}{n}\log\left(\frac{|A|}{n}\right)
=-\sum_{1\leq k\leq n}c^P_k\frac{k}{n}\log\left(\frac{k}{n}\right)$ of partitions $P$ (binary logarithm) enables to measure the distance between $P$ and $Q$ as
\begin{equation}
VI(P,Q)=2e(P\wedge Q)-e(P)-e(Q)\text ,
\end{equation}
Notice that while the range of MMD is $\{0,1,\ldots ,n-1\}\subset\mathbb Z_+$, VI ranges in a finite subset of interval $[0,\log n]\subset\mathbb R_+$. Most importantly,
the entropy $e(P)$ of partitions $P$ is strictly order-inverting, symmetric and submodular, with $e(P_{\bot})=\log(n)$ and $e(P^{\top})=0$. To see submodularity, simply consider
$N=\{1,2,3\}$ as before, and set $P=[12]$ and $Q=[23]$, yielding $P\wedge Q=P_{\bot}$ and $P\vee Q=P^{\top}$. Then, $e(P\vee Q)+e(P\wedge Q)-e(P)-e(Q)=$
\begin{eqnarray*}
&=&-1\log(1)-3\left(\frac{1}{3}\log\left(\frac{1}{3}\right)\right)+2\left(\frac{2}{3}\log\left(\frac{2}{3}\right)+\frac{1}{3}\log\left(\frac{1}{3}\right)\right)=\\
&=&\frac{4}{3}-\log(3)=1.\overline 3-1.585<0.
\end{eqnarray*}
Finally observe that $-e(\cdot)$, in turn, conversely is strictly order-preserving, symmetric and supermodular. It can also be anticipated that VI is in the broad class of metric
distances defined in the sequel, but MMD is not.

\subsection{HD and VI: axioms}
Following \cite{Meila2007}, attention is now placed on those axioms that characterize both partition distance measures HD and VI. An alternative axiomatic characterization
of HD appears in \cite{MirkinCherny1970}. The following proposition may be compared with \cite[pp. 880-881, Property 1]{Meila2007}.
\begin{proposition}
HD is a \textsl{metric}: for all $P,P',Q\in\mathcal P^N$,
\begin{enumerate}
\item $HD(P,Q)=HD(Q,P)$,
\item $HD(P,Q)\geq 0$, with equality if and only if $P=Q$,
\item $HD(P,P')+HD(P',Q)\geq HD(P,Q)$, i.e. triangle inequality.
\end{enumerate}
\end{proposition}

\begin{proof}
The first condition is obvious. In view of lemma 1 above, the second one is also immediate as $\min\{s(P),s(Q)\}\geq s(P\wedge Q)$. In fact, $HD(P,Q)$ is the sum
$[s(P)-s(P\wedge Q)]+[s(Q)-s(P\wedge Q)]$ of two positive integers, while $\underset{P\neq Q}{\min}\text{ }HD(P,Q)=HD(P_{\bot},[ij])=1=s([ij])$ (for any atom $[ij]$). Concerning
triangle inequality, difference $HD(P,P')+HD(P',Q)-HD(P,Q)=$
\begin{equation*}
=2[s(P')-s(P\wedge P')-s(P'\wedge Q)+s(P\wedge Q)]
\end{equation*}
must be shown to be positive for all triplets $P,P',Q\in\mathcal P^N$. For any $P,Q\in\mathcal P^N$, size $s(P\wedge Q)$ is given, and thus $s(P')-[s(P\wedge P')+s(P'\wedge Q)]$
has to be minimized by suitably choosing $P'$. Firstly, sum $s(P\wedge P')+s(P'\wedge Q)$ is maximized when both $P\wedge P'=P$ (or $P'\geqslant P$) and $P'\wedge Q=Q$ (or
$P'\geqslant Q$) hold. Secondly, if $Q\leqslant P'\geqslant P$, then the whole difference is minimized when $P'=P\vee Q$. Thus, HD satisfies triangle inequality as long as the
size satisfies supermodularity: $s(P\vee Q)-s(P)-s(Q)+s(P\wedge Q)\geq 0\text{ for all }P,Q\in\mathcal P^N$. The simplest way to see that this is indeed the case is by focusing on
M\"obius inversion of lattice (or more generally poset) functions (see \cite{Aigner79,RotaMobius} and above). By definition, the size $s(\cdot)$ has M\"obius inversion
$\mu^s:\mathcal P^N\rightarrow\{0,1\}$ given by $\mu^s(P)=1$ if $P$ is an atom (i.e. $P=[ij]\in\mathcal P^N_{\mathcal A}$ or $r(P)=1$), and $\mu^s(P)=0$ otherwise. In fact,
$s(P)=\sum_{Q\leqslant P}\mu^s(Q)$ for all $P\in\mathcal P^N$. The size thus satisfies a sufficient (but not necessary) condition for supermoduarity, in that its M\"obius
inversion takes only positive values (see Section 2.1). This completes the proof.
\end{proof}

Triangle inequality is satisfied with equality by both HD and VI as long as $P'=P\wedge Q$ (for VI, see \cite[pp. 883, 888]{Meila2007} Properties 6 and 10(A.2)).

\begin{proposition}
HD satisfies horizontal collinearity:
\begin{equation*}
HD(P,P\wedge Q)+HD(P\wedge Q,Q)=HD(P,Q)\text{ for all }P,Q\in\mathcal P^N\text .
\end{equation*}
\end{proposition}

\begin{proof}
$HD(P,P\wedge Q)+HD(P\wedge Q,Q)=[s(P)-s(P\wedge Q)]+[s(Q)-s(P\wedge Q)]$ as well as $HD(P,Q)=s(P)+s(Q)-2s(P\wedge Q)$.
\end{proof}

Briefly aticipating the forthcoming analysis, it may be noted that horizontal collinearity may well be conceived in terms of the join, rather than the meet, of any two partitions,
since it is not hard to define distances $d:\mathcal P^N\times\mathcal P^N\rightarrow\mathbb R_+$ satisfying triangle inequality with equality when $P'=P\vee Q$; that is to say,
$d(P,P\vee Q)+d(P\vee Q,Q)=d(P,Q)$ for all $P,Q\in\mathcal P^N$. This is in fact the so-called $B_{\vee}$ ``betweenness'' relation proposed in \cite[p. 176]{Monjardet1981}.

Collinearity also applies to distances between partitions $P,Q$ that are comparable, i.e. either $P\geqslant Q$ or $Q\geqslant P$. Firstly consider the case involving the top
$P^{\top}$ and bottom $P_{\bot}$ elements (for VI, see \cite[p. 888]{Meila2007} property 10(A.1)).

\begin{proposition}
HD satisfies vertical collinearity:
\begin{equation*}
HD(P_{\bot},P)+HD(P,P^{\top})=HD(P_{\bot},P^{\top})\text{ for all }P\in\mathcal P^N\text .
\end{equation*}
\end{proposition}

\begin{proof}
$HD(P_{\bot},P)+HD(P,P^{\top})=s(P)+s(P^{\top})-s(P)=s(P^{\top})$ independently from $P$, as well as $HD(P_{\bot},P^{\top})=s(P^{\top})=\binom{n}{2}$.
\end{proof}

Vertical collinearity may be generalized for arbitrary comparable partitions $P^{\top}\geqslant P>Q\geqslant P_{\bot}$, in that $HD(Q,P')+HD(P',P)=HD(Q,P)$ for all $P'\in[Q,P]$,
where $[Q,P]=\{P':Q\leqslant P'\leqslant P\}$ is an interval or segment \cite{RotaMobius} of $(\mathcal P^N,\wedge ,\vee)$ (see above). In fact, this is precisely the
``interval betweenness'' property considered in \cite[p. 179]{Monjardet1981} (for valuations of distributive lattices).

\section{Distances between complementary partitions}
The distance between the bottom and top elements in vertical collinearity leads to regard such lattice elements as complements, thereby focusing on the distance between other,
generic complements. Maintaining the traditional Hamming distance between subsets as the fundamental benchmark, it must be taken into account that the subset and partition
lattices are very different in terms of complementation. In particular, every subset $A\in 2^N$ has a unique complement $A^c$, and the distance between any two such complements
equals the distance between the bottom and top elements, i.e. $|A\Delta A^c|=n=|N\Delta\emptyset|$ for all $A\in 2^N$. Conversely, partitions $P$ generally have several and quite
different complements \cite{Aigner79}, which are all those $Q$ such that $P\wedge Q=P_{\bot}$ as well as $P\vee Q=P^{\top}$. In statistical classification, partitions $P,Q$
satisfying only the former condition, i.e. $P\wedge Q=P_{\bot}$, are commonly referred to as ``dual partitions'' and investigated as those where the addjusted Rand index ARI
\cite{HubertArabie1985} takes negative values; see \cite[pp. 237-238, 389]{Mirkin1996}, \cite[pp. 429-430]{KovalevaMirkin} and \cite{SchreiderSharov1982}. Apart from this,
concerning complementation and partition distances MMD, VI and HD, the former measures the distance between any two complements $P,Q$ solely through their cardinalities $|P|,|Q|$,
while VI and HD provide a fine distinction between different complements, and also agree on which are closer and which are remoter. The issue may be exemplified with
$N=\{1,\ldots ,7\}$ and partitions $P=123|456|7$ and $P^*=147|2|3|5|6$ and $P_*=1|2|34|5|67$ (where vertical bar $|$ separates blocks). Both $P^*$ and $P_*$ are complements of
$P$, that is $P\wedge P^*=P\wedge P_*=P_{\bot}$ and $P\vee P^*=P\vee P_*=P^{\top}$. Distances MMD, VI and HD are:
\begin{eqnarray*}
MMD(P,P_*)=4&=&4=MMD(P,P^*)\text ,\\
VI(P,P_*)=\frac{6\log 6-2}{7}\simeq 1.93&<&1.95\simeq\frac{4\log 9+2\log 3-1}{7}=VI(P,P^*)\text ,\\
HD(P,P_*)=8&<&9=HD(P,P^*)\text .
\end{eqnarray*}
Concerning MMD, this examples generalizes as follows.

\begin{proposition}
For any two complementary partitions $P,Q\in\mathcal P^N$,
\begin{equation}
MMD(P,Q)=\max\{r(P),r(Q)\}\text .
\end{equation}
\end{proposition}

\begin{proof}
If $P\wedge Q=P_{\bot}$, then every edge $\{A,B\}\in E\subset P\times Q$ of the bipartite graph $G=(P\cup Q,E)$ defined in Section 2 above has unit weight $1=|A\cap B|$. Hence, a
maximum-weight matching simply is one including the maximum number of feasible edges. Such a number is $\sum_{A\in P\vee Q}\min\{|P^A|,|Q^A|\}$, because each block (of either
partition) can be the endpoint of at most one edge included in a matching. Also, the number of elements $i\in N$ that must be deleted for the two residual partitions to coincide
is $\sum_{A\in P\vee Q}(|A|-\min\{|P^A|,|Q^A|\})$. On the other hand, $P\vee Q=P^{\top}$ entails
\begin{equation*}
\sum_{A\in P\vee Q}(|A|-\min\{|P^A|,|Q^A|\})=n-\min\{|P|,|Q|\}=\max\{r(P),r(Q)\}
\end{equation*}
as desired.
\end{proof}

As shown by the above example, a partition generally has different complements with different classes. The set of complements of any partition $P$ is denoted by
$\mathcal{CO}(P)=\{Q:P\wedge Q=P_{\bot},P\vee Q=P^{\top}\}$.
A \textit{modular element} of the partition lattice \cite{Aigner79,Stanley1971,Stern99} is any $P\in\mathcal P^N$ where all blocks are singletons apart from only one, at most,
i.e. $\sum_{1<k\leq n}c_k(P)\leq 1$. The sublattice $\mathcal P^N_{mod}\subseteq\mathcal P^N$ consisting of modular elements contains the bottom and top elements, together with
all partitions of the form $\{A\}\cup P_{\bot}^{A^c}$ with $1<|A|<n$, where $P^{A^c}_{\bot}$ is the finest partition of $A^c$. Hence there are $2^n-n$ modular partitions (with
$\mathcal P^N_{mod}=\mathcal P^N$ for $n\leq 3$).
Here, the main link between modular elements and complementation is that an element is modular if and only if no two of its complements are comparable
\cite[Theorem 1]{Stanley1971}. Therefore, if $P\not\in\mathcal P^N_{mod}$, then there are $Q,Q'\in\mathcal{CO}(P)$ such that $Q>Q'$. It seems thus important that the distance
between $P$ and $Q$ differs from the distance between $P$ and $Q'$. The following result bounds the Hamming distance HD between a partition and any of its complements. 

\begin{proposition}
For all $P\in\mathcal P^N$, if $Q\in\mathcal{CO}(P)$, then
\begin{equation*}
s(P)+|P|-1\leq HD(P,Q)\leq s(P)+\binom{|P|}{2}\text ,
\end{equation*}
where the upper bound is always tight, while the lower one is tight only if
\begin{equation*}
c_1(P)\leq 2+\sum_{1<k\leq n}(k-2)c_k(P)\text .
\end{equation*}
\end{proposition}

\begin{proof}
Firstly note that if $Q\in\mathcal{CO}(P)$, then $HD(P,Q)=s(P)+s(Q)$. Hence,
\begin{equation*}
\min\{s(Q):Q\in\mathcal{CO}(P)\}\leq HD(P,Q)-s(P)\leq\max\{s(Q):Q\in\mathcal{CO}(P)\}\text .
\end{equation*}
Any complement of partition $P=\{A_1,\ldots ,A_{|P|}\}$ has join-decompositions minimally involving $|P|-1$ atoms $[ij]_1,\ldots ,[ij]_{|P|-1}\in\mathcal P^N_{\mathcal A}$, with
associated pairs $\{i,j\}_m\in N_2$ such that $|A_m\cap\{i,j\}_m|=1=|A_{m+1}\cap\{i,j\}_m|,1\leq m<|P|$. Considering the upper bound first, observe that size
$s([ij]_1\vee\cdots\vee[ij]_{|P|-1})$ attains its maximum when $|\{i,j\}_m\cap\{i,j\}_{m+1}|=1$ for all $1\leq m<|P|-1$, in which case
$s([ij]_1\vee\cdots\vee [ij]_m)=\binom{m+1}{2}$ for all $1\leq m<|P|$. This bound is tight because such a complement $P^*=[ij]_1\vee\cdots\vee[ij]_{|P|-1}$ always exists, whatever
the class $c(P)$ of $P$. In fact, $P^*\in\mathcal P^N_{mod}$ has $n-|P|+1$ blocks, out of which $n-|P|$ are singletons, while the remaining one $B\in P^*$ is $|P|$-cardinal and
satisfies $|B\cap A|=1$ for all $A\in P$, i.e. $P^*=\{B\}\cup P_{\bot}^{B^c}$. Thus $s(P^*)=\binom{|P|}{2}$.

Turning to the lower bound, observe that size $s([ij]_1\vee\cdots\vee[ij]_{|P|-1})$ attains its minimum, ideally, when $\{i,j\}_m\cap\{i,j\}_{m'}=\emptyset$ for all
$1\leq m<m'<|P|$, in which case $s([ij]_1\vee\cdots\vee [ij]_m)=m$ for all $1\leq m<|P|$. Yet, this is not always possible because each block $A\in P$ can have non-empty
intersection with a number of pair-wise disjoint pairs $\{i,j\}_m,1\leq m<|P|$ which is bounded above by $|A|$, entailing that the constraint is given by the number $c_1(P)$
of singletons $\{i\}\in P$. Specifically, nesting together $\sum_{1<k\leq n}c_k(P)$ non-singleton blocks requires $\sum_{1<k\leq n}c_k(P)-1$ pairs $\{i,j\}_m$. If these latter
have to be pair-wise disjoint, then the maximum number of elements $j\in N$ in non-singleton blocks available to match (into pair-wise disjoint pairs) those elements
$\{i\}\in P$ in singletons is $\sum_{1<k\leq n}kc_k(P)-2\left(\sum_{1<k\leq n}c_k(P)-1\right)$.
\end{proof}

In words, if the number $c^P_1$ of singleton blocks of partition $P$ exceeds the number $2+\sum_{1<k\leq n}(k-2)c_k(P)$ of elements $j\in N$ available to match, into pair-wise
disjoint pairs, those elements $\{i\}\in P$ in singletons, then basically a complement $Q$ of $P$, in order to yield the top partition $P^{\top}$ through the join $P\vee Q$, must
necessarily consist of blocks larger than pairs. Of course, in the limit, if the blocks of $P=P_{\bot}$ are all singletons, then the unique complement $Q=P^{\top}$ has to be the
coarsest or top partition, consisting of a unique block. Thus, the greater the number $c^P_1$ of singleton blocks of $P$, the fewer and larger the blocks $B\in Q$ of a complement
$Q$ of $P$ have to be. In this view, the following Proposition 6 shows that the more the cardinality $|B|$ of these blocks $B\in Q$ is evenly distributed, the lower the size
$s(Q)$ of the complement $Q$. In fact, on any level $\mathcal P^N_{k}=\{P:P\in\mathcal P^N,|P|=n-k\},0\leq k<n$ of the partition lattice, the size attains its maximum value on
modular partitions (consisting of $n-k-1$ singletons and one $k+1$-cardinal block) and its minimum value on those partitions each of whose block has cardinality between
$\lfloor\frac{n}{n-k}\rfloor$ and $\lceil\frac{n}{n-k}\rceil$, where $\lfloor\alpha\rfloor$ is the floor of $\alpha$, i.e. the greatest integer $\leq\alpha$, while
$\lceil\alpha\rceil$ is the ceiling of $\alpha$, i.e. the smallest integer $\geq\alpha$, for $\alpha\in\mathbb R_+$. As detailed by Proposition 8 next, the opposite occurs for the
entropy of partitions.

\begin{proposition}
If $P\in\mathcal P^N$ satisfies $2+\sum_{1<k\leq n}(k-2)c_k(P)<c_1(P)$, then
\begin{eqnarray*}
\underset{P_*\in\mathcal{CO}(P)}{\min} s(P_*)&=&\left[\theta(P)\left(\left\lfloor\frac{n}{\theta(P)}\right\rfloor+1\right)-n\right]
\binom{\left\lfloor\frac{n}{\theta(P)}\right\rfloor}{2}\\
&+&\left(n-\theta(P)\left\lfloor\frac{n}{\theta(P)}\right\rfloor\right)\binom{\left\lceil\frac{n}{\theta(P)}\right\rceil}{2}\text ,
\end{eqnarray*}
where $\theta(P)=1+\sum_{1<k\leq n}c_k(P)(k-1)$.
\end{proposition}

\begin{proof}
If $2+\sum_{1<k\leq n}(k-2)c_k(P)<c_1(P)$, then the above proof of Proposition 5 entails that the maximum number $\max\{|Q|:Q\in\mathcal{CO}(P)\}$ of blocks of a complement of
$P$ is $\theta(P):=1+\sum_{1<k\leq n}c_k(P)(k-1)$. On the other hand, for $0<m\leq n$, among $m$-cardinal partitions $Q$ of a $n$-set the size is minimized when
$|B|\in\left\{\left\lfloor\frac{n}{m}\right\rfloor,\left\lceil\frac{n}{m}\right\rceil\right\}$ for all $B\in Q$. Bound $\underset{P_*\in\mathcal{CO}(P)}{\min} s(P_*)$ above is
the size of a $\theta(P)$-cardinal partition $P_*$ with $|B|\in\left\{\left\lfloor\frac{n}{\theta(P)}\right\rfloor,\left\lceil\frac{n}{\theta(P)}\right\rceil\right\}$ for all
$B\in P_*$. In particular, the number of $\left\lfloor\frac{n}{\theta(P)}\right\rfloor$-cardinal blocks is
$\theta(P)\left(\left\lfloor\frac{n}{\theta(P)}\right\rfloor+1\right)-n$, while the number of $\left\lceil\frac{n}{\theta(P)}\right\rceil$-cardinal blocks is
$n-\theta(P)\left\lfloor\frac{n}{\theta(P)}\right\rfloor$.
\end{proof}
 
\begin{proposition}
Among complements $Q\in\mathcal{CO}(P)$ of any $P\in\mathcal P^N$, HD and VI have common minimizers, i.e.
$\underset{Q\in\mathcal{CO}(P)}{\arg\min}\text{ }HD(P,Q)=\underset{Q\in\mathcal{CO}(P)}{\arg\min}\text{ }VI(P,Q)$,
and common maximizers, i.e. $\underset{Q\in\mathcal{CO}(P)}{\arg\max}\text{ }HD(P,Q)=\underset{Q\in\mathcal{CO}(P)}{\arg\max}\text{ }VI(P,Q)$.
\end{proposition}

\begin{proof}
Firstly, $Q\in\mathcal{CO}(P)$ entails $VI(P,Q)=2\log n-e(P)-e(Q)$. Thus, $VI(P,Q)$ is minimized or else maximized when $e(Q)$ is, respectively, maximized or else minimized. On
the other hand, if $P\in\mathcal P^N_{mod}$, then all complements $Q\in\mathcal{CO}(P)$ have same rank. Otherwise, as already observed, there are comparable complements, i.e. with
different rank. Therefore, in general, among complements $Q\in\mathcal{CO}(P)$ entropy $e(Q)$ is minimized when $|Q|$ is minimized and, in addition, $Q\in\mathcal P^N_{mod}$. This
is precisely where size $s(Q)$ is maximized. Similarly, $e(Q)$ is maximized when $|Q|$ is maximized and, in addition,
$|B|\in\left\{\left\lfloor\frac{n}{|Q|}\right\rfloor,\left\lceil\frac{n}{|Q|}\right\rceil\right\}$ for all $B\in Q$. This is where $s(Q)$ is minimized.
\end{proof}

\section{Minimum-weight paths between partitions}
This section provides an analysis similar, in spirit, to that provided in \cite[Section 3]{Monjardet1981}, although the generic posets and semilattices considered there are replaced
here with the geometric lattice of partitions. Similarly, the covering graph becomes the graph $\mathbb G$ of polytope $\mathbb P_N$ below, and despite posets lack the join and meet
operators, still \cite{Monjardet1981} defines upper/lower valuations, which correspond to sub/supermodular partition functions in the present setting. Apart from these
differences, still the general idea to define metrics through weighted paths in the graph induced by the covering relation is the same. In fact, 
Hamming distance $|E\Delta E'|$ between edge sets $E,E'\in 2^{N_2}$ is the length of a shortest path between vertices $\chi_E,\chi_{E'}\in\{0,1\}^{\binom{n}{2}}$ of the
$\binom{n}{2}$-dimensional unit hypercube $[0,1]^{\binom{n}{2}}$, where $\chi_E:N_2\rightarrow\{0,1\}$ is the characteristic function defined in Section 2, i.e.
$\chi_E(\{i,j\})=1$ if $\{i,j\}\in E$ and 0 otherwise. Recall that a polytope naturally defines a graph with its same vertices and edges \cite[p. 93]{BronstedConvex}, and the
hypercube is perhaps the main example of polytope. In particular, the graph of hypercube $[0,1]^{\binom{n}{2}}$ is the Hasse diagram of Boolean lattice $(2^{N_2},\cap ,\cup)$, for its
edges correspond to the covering relation, that is to say $\{E,E'\}$ is an edge of the hypercube if either $E\supset E',|E|=|E'|+1$ or else the converse, i.e.
$E'\supset E,|E'|=|E|+1$.

Clearly, a shortest path is a minimum-weight path as long as every edge has unit weight. This simple observation is the starting point toward an analog view of the Hamming
distance HD between partitions, namely as the weight of a minimum-weight path in the associated Hasse diagram when edge weights are determined by the size. More generally, if edge
weights are determined by a symmetric and strictly order preserving/inverting partition function (like rank or entropy), then minimum-weight paths across edges of the
Hasse diagram equivalently yield well-defined metric distances.
In this view, consider the convex hull $co.hu(\{I_P:P\in\mathcal P^N\})=\mathbb P_N$ whose extreme points
\cite{BronstedConvex,Branko2001} are all the $\mathcal B_n$ Boolean $\binom{n}{2}$-vectors defined by the indicator functions $I_P,P\in\mathcal P^N$ of partitions. Note that
$\mathbb P_N$ is a (so-called ``hull onest'') 0/1-polytope that might be included in the classifying literature \cite{Aicholzer+96,Ziegler2000} as a type in and of itself. Here, it may be
referred to as ``the polytope of partitions'', since its graph $\mathbb G=(\mathcal P^N,\mathbb E)$ basically is the Hasse diagram of partition lattice $(\mathcal P^N,\wedge,\vee)$.
Specifically, edges correspond to the covering relation between partitions, i.e. $\{P,Q\}\in\mathbb E$ if either $[Q,P]=\{P,Q\}$ or else $[P,Q]=\{P,Q\}$ (see above). Let
$P\gtrdot Q\Leftrightarrow[Q,P]=\{P,Q\}$ denote the covering relation between partitions, while $ex(\mathbb P_N)=\{I_P:P\in\mathcal P^N\}$ is the set of extreme points or
vertices of $\mathbb P_N$. For $N=\{1,2,3\}$, polytope $\mathbb P_N$ is strictly included in $[0,1]^3$ and its five vertices are $(0,0,0)$, $(1,0,0)$, $(0,1,0)$, $(0,0,1)$ and
$(1,1,1)$. Thus, vertices $(1,1,0)$, $(1,0,1)$ and $(0,1,1)$ of $[0,1]^3$ are excluded from $ex(\mathbb P_N)$, as they correspond to those 3 graphs with vertex set
$\{1,2,3\}$ whose edge set is 2-cardinal. That is, $2^{\binom{n}{2}}-\mathcal B_n$ is precisely the number of graphs with vertex set $N$ that do not coincide with their closure
(see Sections 1 and 2). Geometrically, for $N=\{1,2,3\}$, polytope $\mathbb P_N$ is the union of a lower tetrahedron, whose volume is $0.1\overline 6$, and an upper tetrahedron,
whose volume is $0.\overline 3$, hence the whole volume is $0.5$. The former is $co.hu((0,0,0),(1,0,0),(0,1,0),(0,0,1))$, while the latter is $co.hu((1,0,0),(0,1,0),(0,0,1),(1,1,1))$.
Thus $\mathbb P_N$ is the polyhedron obtained as the intersection of 6 half-spaces, namely those three above the hyperplanes each including one of the three facets (different from
unit simplex $co.hu((1,0,0),(0,1,0),(0,0,1))$) of the lower tetrahedron, and those three below the hyperplanes each including one of the three facets (again different from the unit
simplex) of the upper tetrahedron.
Although this situation for $N=\{1,2,3\}$ is quite simple, still for generic $N=\{1,\ldots ,n\}$ the associated polytope $\mathbb P_N$ is more complex. When $n=4$, for example,
$\mathbb P_N\subset[0,1]^6$ is the convex hull of the 15 vertices corresponding to the rows of Table 2 below, with columns indexed, from left to right, by the $\binom{4}{2}=6$ atoms
[12], [13], [14], [23], [24] and [34] of $\mathcal P^N$. Corresponding partitions are in the far left column, with vertical bar $|$ separating blocks.

\begin{table}[htbp]
\caption{\textsl{Extreme points of 0/1-polytope $\mathbb P_N\subset[0,1]^{\binom{4}{2}}$ for $N=\{1,2,3,4\}$}}
\label{tab: extreme points}
\smallskip
\begin{center}
\begin{tabular}{|c|c|c|c|c|c|c|}
\hline
$P\in\mathcal P^N$ $\downarrow$ ; $[ij]\in\mathcal P^N_{\mathcal A}$ $\rightarrow$ & [12] & [13] & [14] & [23] & [24] & [34]\\
\hline
$P\bot=1|2|3|4$ &0 & 0 & 0 & 0 & 0 & 0 \\
\hline
$[12]=12|3|4$ & 1 & 0 & 0 & 0 & 0 & 0\\
\hline
$[13]=13|2|4$ & 0 & 1 & 0 & 0 & 0 & 0\\
\hline
$[14]=14|2|3$ & 0 & 0 & 1 & 0 & 0 & 0\\
\hline
$[23]=1|23|4$ & 0 & 0 & 0 & 1 & 0 & 0\\
\hline
$[24]=1|24|3$ & 0 & 0 & 0 & 0 & 1 & 0\\
\hline
$[34]=1|2|34$ & 0 & 0 & 0 & 0 & 0 & 1\\
\hline
$12|34$ & 1 & 0 & 0 & 0 & 0 & 1\\
\hline
$13|24$ & 0 & 1 & 0 & 0 & 1 & 0\\
\hline
$14|23$ & 0 & 0 & 1 & 1 & 0 & 0\\
\hline
$123|4$ & 1 & 1 & 0 & 1 & 0 & 0\\
\hline
$124|3$ & 1 & 0 & 1 & 0 & 1 & 0\\
\hline
$134|2$ & 0 & 1 & 1 & 0 & 0 & 1\\
\hline
$1|234$ & 0 & 0 & 0 & 1 & 1 & 1\\
\hline
$P\top=1234$ & 1 & 1 & 1 & 1 & 1 & 1\\
\hline
\end{tabular}
\end{center}
\end{table}

As for weights on edges $\{P,Q\}\in\mathbb E$ of (covering) graph $\mathbb G$, let $\mathbb F\subset\mathbb R^{\mathcal B_n}$ be the vector space of strictly
order-preserving/inverting and symmetric partition functions $f:\mathcal P^N\rightarrow\mathbb R$. As already mentioned entropy, rank and size are in $\mathbb F$, and the former
is order-inverting, while the latter two are order-preserving. Given any $f\in\mathbb F$, define weights
$w_f:\mathbb E\rightarrow\mathbb R_{++}$ on edges $\{P,Q\}\in\mathbb E$ by
\begin{equation*}
w_f(\{P,Q\})=\max\{f(P),f(Q)\}-\min\{f(P),f(Q)\}\text .
\end{equation*}
For all pairs $P,Q\in\mathcal P^N$, let $Path(P,Q)$ contain all $P-Q$-paths in graph $\mathbb G$, noting that this latter is highly connected (or dense), as every partition
$P$ is covered by $\binom{|P|}{2}$ partitions $Q$ and covers $\sum_{A\in P}\left(2^{|A|-1}-1\right)$ partitions $Q'$, hence $|Path(P,Q)|\gg1$ for all $P,Q$. Recall that a path
$p(P,Q)\in Path(P,Q)$ is a subgraph $p(P,Q)=(V^p_{P,Q},E^p_{P,Q})\subset\mathbb G$ where
\begin{eqnarray*}
V^p_{P,Q}&=&\{P=P_0,P_1,\ldots ,P_m=Q\}\text{ and}\\
E^p_{P,Q}&=&\{\{P_0,Q_0\},\{P_1,Q_1\},\ldots ,\{P_{m-1},Q_{m-1}\}\}\text ,
\end{eqnarray*}
with $P_{k+1}=Q_k$ for $0\leq k<m$. Also, the weight of a path $p(P,Q)$ is
\begin{equation*}
w_f(p(P,Q))=\sum_{0\leq k<m}w_f(\{P_k,Q_k\})\text .
\end{equation*}

\begin{definition}
Minimum-$f$-weight partition distance $\delta_f:\mathcal P_N\times\mathcal P_N\rightarrow\mathbb R_+$ is
\begin{equation}
\delta_f(P,Q):=\underset{p(P,Q)\in Path(P,Q)}{\min}w_f(p(P,Q))\text{ for all }f\in\mathbb F\text .
\end{equation}
\end{definition}

\begin{proposition}
For all $f\in\mathbb F$ and all $P,Q\in\mathcal P^N$, every minimum-$f$-weight $P-Q$-path visits $P\wedge Q$ or $P\vee Q$ or both; that is to say, if path $p(P,Q)$ satisfies
$w_f(p(P,Q))=\delta_f(P,Q)$, then $V^p_{P,Q}\cap\{P\wedge Q,P\vee Q\}\neq\emptyset$.
\end{proposition}

\begin{proof}
If $P,Q$ are comparable, say $P\geqslant Q$, then $\{P\vee Q,P\wedge Q\}\subseteq V^p_{P,Q}$ for \textit{all} paths $p(P,Q)\in Path(P,Q)$, in that $P=P\vee Q$ and $Q=P\wedge Q$;
in particular, if $P\gtrdot Q$, then the unique minimum-$f$-weight $P-Q$ path consists of vertices $P$ and $Q$ together with the edge $\{P,Q\}\in\mathbb E$ linking them. On the
other hand, if $P,Q$ are not comparable, i.e. $P\not\geqslant Q\not\geqslant P$, then any path $p(P,Q)$ visits some vertex $P'$ comparable with both $P,Q$, and either $P'>P,Q$ or
else $P,Q>P'$. Hence $p(P,Q)=p(P,P')\cup p(P',Q)$, with $E^p_{P,P'}\cap E^p_{P',Q}=\emptyset$, for some $P-P'$-path $p(P,P')$ and some $P'-Q$-path $p(P',Q)$, entailing that the weight
of such a $p(P,Q)$ is $w_f(p(P,Q))=w_f(p(P,P'))+w_f(p(P',Q))$. Finally, since $f$ is strictly order-preserving/inverting and symmetric, $P'=P\vee Q$ minimizes
$w_f(p(P,P'))+w_f(p(P',Q))$ over all partitions $P'>P,Q$ while $P'=P\wedge Q$ minimizes $w_f(p(P,P'))+w_f(p(P',Q))$ over all $P'<P,Q$.
\end{proof}

Whether a minimum-$f$-weight path visits the join or else the meet of any two incomparable partitions clearly depends on $f$. A generic $f\in\mathbb F$ may have associated
minimum-weight paths visiting the meet of some incomparable partitions $P,Q$ and the join of some others $P',Q'$. In fact, whether minimum-weight paths awlays visit the meet or
else the join of any two incomparable partitions depends on whether $f$ or else $-f$ is supermodular. As already observed, if $f$ is supermodular, then $-f$ is submodular, i.e.
$-f(P\wedge Q)-f(P\vee Q)\leq -f(P)-f(Q)$ (and viceversa).

\begin{proposition}
For any strictly order-preserving $f\in\mathbb F$, if $f$ is supermodular, then the minimum-$f$-weight partition distance is
\begin{equation*}
\delta_f(P,Q)=f(P)+f(Q)-2f(P\wedge Q)\text ,
\end{equation*}
while if $f$ is submodular, then the minimum-$f$-weight partition distance is
\begin{equation*}
\delta_f(P,Q)=2f(P\vee Q)-f(P)-f(Q)\text .
\end{equation*}
\end{proposition}

\begin{proof}
Supermodularity entails
\begin{equation*}
2f(P\vee Q)-f(P)-f(Q)\geq f(P\vee Q)-f(P\wedge Q)\geq f(P)+f(Q)-2f(P\wedge Q)\text ,
\end{equation*}
whereas submodularity entails
\begin{equation*}
2f(P\vee Q)-f(P)-f(Q)\leq f(P\vee Q)-f(P\wedge Q)\leq f(P)+f(Q)-2f(P\wedge Q)\text ,
\end{equation*}
for all $P,Q\in\mathcal P^N$.
\end{proof}

\begin{proposition}
For any strictly order-inverting $f\in\mathbb F$, if $f$ is supermodular, then the minimum-$f$-weight partition distance is
\begin{equation*}
\delta_f(P,Q)=f(P)+f(Q)-2f(P\vee Q)\text ,
\end{equation*}
while if $f$ is submodular, then the minimum-$f$-weight partition distance is
\begin{equation*}
\delta_f(P,Q)=2f(P\wedge Q)-f(P)-f(Q)\text .
\end{equation*}
\end{proposition}

\begin{proof}
Supermodularity entails
\begin{equation*}
2f(P\wedge Q)-f(P)-f(Q)\geq f(P\wedge Q)-f(P\vee Q)\geq f(P)+f(Q)-2f(P\vee Q)\text ,
\end{equation*}
whereas submodularity entails
\begin{equation*}
2f(P\wedge Q)-f(P)-f(Q)\leq f(P\wedge Q)-f(P\vee Q)\leq f(P)+f(Q)-2f(P\vee Q)\text ,
\end{equation*}
for all $P,Q\in\mathcal P^N$.
\end{proof}

Since the size $s$ is supermodular (see Proposition 1) and order-preserving, $HD$ is the minimum-$s$-weight partition distance, i.e. $HD(P,Q)=\delta_s(P,Q)$ for all $P,Q$. On the
other hand, the rank $r$ is submodular \cite[pp. 259, 265, 274]{Aigner79} and order-preserving, hence $\delta_r(P,Q)=2r(P\vee Q)-r(P)-r(Q)=|P|+|Q|-2|P\vee Q|$ is the
minimum-$r$-weight partition distance. In particular, $w_r(\{P,Q\})=1$ for all edges $\{P,Q\}\in\mathbb E$, and therefore $\delta_r$ is in fact the shortest-path distance. This
is detailed below by means of Example 2. Finally, entropy $e$ is order-inverting and submodular, hence the minimum-$e$-weight distance $\delta_e$ is the VI distance
$VI(P,Q)=2e(P\wedge Q)-e(P)-e(Q)$, as shown in Example 1 hereafter. Propositions 9 and 10 are summarized in Table 3 below.

\begin{example}
\textbf{Entropy-based minimum-weight path distance:} for any two atoms $[ij],[ij']\in\mathcal P^N_{\mathcal A}$ such that $\{i,j\}\cap\{i,j'\}=\{i\}$, the VI distance is
\begin{equation*}
VI([ij],[ij'])=2e([ij]\wedge [ij'])-e([ij])-e([ij'])=2\log n-2\left(\log n-\frac{2}{n}\right)=\frac{4}{n}\text ,
\end{equation*}
and this is indeed the minimum-$e$-weight distance. On the other hand,
\begin{equation*}
e([ij])+e([ij'])-2e([ij]\vee[ij'])=2\left(\log n-\frac{2}{n}\right)-2\left(\log n-\frac{3}{n}\log 3\right)=
\end{equation*}
$=\frac{2}{n}(3\log 3-2)$. In fact, $\frac{4}{n}=VI([ij],[ij'])<\frac{2}{n}(3\log 3-2)$ as $4<3\log 3$.
\end{example}

\begin{example}
\textbf{Rank-based shortest path distance:} let $N=\{1,2,3,4,5,6,7\}$ and consider partitions $P=135|27|46$ and $Q=1|23|47|56$ (with vertical bar $|$ separating blocks as in Table
2 above). Then, $P\wedge Q=1|2|3|4|5|6|7=P_{\bot}$ as well as $P\vee Q=1234567=P^{\top}$. Accordingly,
\begin{equation*}
\delta_r(P,Q)=2r(P\vee Q)-r(P)-r(Q)=|P|+|Q|-2|P\vee Q|=3+4-2=5
\end{equation*}
\begin{equation*}
\text{while }r(P)+r(Q)-2r(P\wedge Q)=2|P\wedge Q|-|P|-|Q|=14-3-4=7
\end{equation*}
as $|P|+|Q|-2|P\vee Q|=5$ is the length of a shortest path between $P$ and $Q$. Such a path visits $P\vee Q=P^{\bot}$ and for instance may be across edges
\begin{equation*}
\{P,12357|46\},\{12357|46,P^{\top}\},\{P^{\top},123|4567\},\{123|4567,1|23|4567\}
\end{equation*}
and finally $\{1|23|4567,Q\}$ of $\mathbb P_N$ (or equivalently of Hasse diagram $\mathbb G$ introduced above). On the other hand, a shortest $P-Q$-path \textit{forced to visit}
$P\wedge Q=P_{\bot}$ has length 7 and for instance may be across edges
\begin{equation*}
\{P,1|35|27|46\},\{1|35|27|46,1|2|35|46|7\},\{1|2|35|46|7,1|2|3|46|5|7\},
\end{equation*}
\begin{equation*}
\{1|2|3|46|5|7,P_{\bot}\},\{P_{\bot},1|23|4|5|6|7\},\{1|23|4|5|6|7,1|23|47|5|6\}
\end{equation*}
and finally $\{1|23|47|5|6,Q\}$. Note that the rank assigns to every edge $\{P,Q\}$ of $\mathbb P_N$ unit weight $w_r(P,Q)=1$, and thus $\delta_r$ is indeed the shortest path distance.
\end{example}

\begin{table}[htbp]
\caption{\textsl{$\delta_f(P,Q)$ for $f$ symmetric, strictly order preserving/inverting, super/submodular.}}
\label{tab: minimum-weight distances}
\begin{center}
\begin{tabular}{|c|c|c|c|}
\hline
$f$ symmetric & $f$ strictly order-preserving & $f$ strictly order-inverting\\
\hline
$f$ supermodular& $f(P)+f(Q)-2f(P\wedge Q)$ & $f(P)+f(Q)-2f(P\vee Q)$\\
\hline
$f$ submodular & $2f(P\vee Q)-f(P)-f(Q)$ & $2f(P\wedge Q)-f(P)-f(Q)$\\
\hline
\end{tabular}
\end{center}
\end{table}

\section{Distinctions, co-atoms and fields}
A further measure of partition entropy, called logical entropy, has been recently proposed \cite{EllermanLogicalEntropy} in terms of distinctions, i.e. \textit{ordered} pairs
$(i,j)\in N\times N$ (see Section 2). In statistical classification, the same concept is also referred to as the ``Gini coefficient''
\cite[pp. 53-54, 247-250, 257, 334]{Mirkin2013}. If distinctions are replaced with unordered pairs $\{i,j\}\in N_2$, then \textit{mutatis mutandis} the non-normalized logical
entropy of partitions $P$ is the analog of $\binom{n}{2}-s(P)$, providing a further minimum-weight partition distance. Furthermore, since in information theory partitions are
generally evaluated by means of order-inverting functions, the approach developed thus far may be applied to the upside-down Hasse diagram of the partition lattice, with co-atoms
(or dual atoms \cite{RotaMobius}) in place of atoms. In this way, the distance between partitions is the distance between the associated fields of subsets.

A partition $P$ distinguishes between $i\in N$ and $j\in N\backslash i$ if $i\in A\in P$ while $j\in B\in P$ with $A\neq B$, and the set of such distinctions has been
recently proposed as the logical analog of the complement of $P$, with the (normalized) number of distinctions providing a novel measure of the (logical) entropy of partitions
\cite{EllermanLogicalEntropy,EllermanLogic}. In particular, this is achieved through apartness binary relations $\mathcal R^c$, wich are the complement of equivalence relations
$\mathcal R$ (see again Section 2). In terms of atoms $[ij]\in\mathcal P^N_{\mathcal A}$ of the partition lattice, the logical entropy $h:\mathcal P^N\rightarrow\mathbb R_+$ of
partitions \cite[p. 127]{EllermanLogicalEntropy} is
\begin{equation}
h(P)=\frac{2|\{[ij]:P\not\geqslant [ij]\}|}{n^2}=\frac{2\left(\binom{n}{2}-s(P)\right)}{n^2}=\frac{n(n-1)-2s(P)}{n^2}\text ,
\end{equation}
with $h(P^{\top})=0=s(P_{\bot})$ and $h(P_{\bot})=\frac{n-1}{n}=\frac{2s(P^{\top})}{n^2}$.

\begin{proposition}
The logical entropy-based minimum-weight distance $\delta_h$ is
\begin{equation*}
\delta_h(P,Q)=2h(P\wedge Q)-h(P)-h(Q)\text{ for all }P,Q\in\mathcal P^N\text .
\end{equation*}
\end{proposition}

\begin{proof}
Logical entropy $h$ satisfies $h\in\mathbb F$ and is strictly order-inverting. Also, apart from constant terms, $h$ varies with $-s$, which is submodular because $s$ is
supermodular. That is to say,
\begin{eqnarray*}
h(P)+h(Q)&=&\frac{2}{n}\left(n-1-\frac{s(P)+s(Q)}{n}\right)\text ,\\
h(P\wedge Q)+h(P\vee Q)&=&\frac{2}{n}\left(n-1-\frac{s(P\wedge Q)+s(P\vee Q)}{n}\right)\text .
\end{eqnarray*}
Thus $s(P\wedge Q)+s(P\vee Q)\geq s(P)+s(Q)\Rightarrow h(P\wedge Q)+h(P\vee Q)\leq h(P)+h(Q)$ and the desired conclusion follows from Proposition 10.
\end{proof}

A field of subsets is a set system $\mathcal F\subseteq 2^N$ closed under union, intersection and complementation, i.e. $A\cap B,A\cup B,A^c\in\mathcal F$ for all
$A,B\in\mathcal F$. Every partition $P\in\mathcal P^N$ generates the field $\mathcal F_P=2^P$ containing all subsets $B\in 2^N$ obtained as the union of blocks $A\in P$, with
$\mathcal F_{P_{\bot}}=2^N$ and $\mathcal F_{P^{\top}}=\{\emptyset,N\}$. There are $2^{n-1}-1$ minimal fields (generated by partitions) that strictly include
$\mathcal F_{P^{\top}}$; they are those $\mathcal F_A=\mathcal F_{A^c}=\{\emptyset,A,A^c,N\}$ with $\emptyset\subset A\subset N$.
On the other hand,
2-cardinal partitions $\{A,A^c\}\in\mathcal P^N$ are the co-atoms \cite{Aigner79} (or dual atoms \cite{RotaMobius}) of partition lattice $(\mathcal P^N,\wedge,\vee)$ ordered
by coarsening. In fact, in information theory finer partitions are generally more valuable than coarser ones, and thus attention is placed on order-inverting partition functions.
In this view, the partition lattice is often dealt with as ordered by refinement and thus with the upside-down Hasse diagram. Accordingly,
a distance between partitions also obtains by counting co-atoms rather than atoms. To this end, define the co-size
$cs:\mathcal P^N\rightarrow\mathbb Z_+$ by $cs(P)=|\{\{A,A^c\}:P\leqslant \{A,A^c\}\}|$, with $cs(P_{\bot})=2^{n-1}-1$ and $cs(P^{\top})=0$. In words, $cs(P)$ is
the number of co-atoms coarser than $P$. 

\begin{proposition}
The minimum-$cs$-weight partition distance is
\begin{equation*}
\delta_{cs}(P,Q)=cs(P)+cs(Q)-2cs(P\vee Q)\text{ for all }P,Q\in\mathcal P^N\text .
\end{equation*}
\end{proposition}

\begin{proof}
Denote by $\hat\mu^{cs}:\mathcal P^N\rightarrow\mathbb Z$ the M\"obius inversion from above \cite{Aigner79,RotaMobius} of the co-size, with
$cs(P)=\sum_{Q\geqslant P}\hat\mu^{cs}(Q)$ for all $P$. By definition, $\hat\mu^{cs}(P)=1$ if $|P|=2$ and 0 otherwise. Like for the size in Proposition 1, this entails
supermodularity, i.e. $cs(P\wedge Q)+cs(P\vee Q)\geq cs(P)+cs(Q)$. Furthermore, $cs\in\mathbb F$ is order-inverting. Therefore,
\begin{equation*}
cs(P)+cs(Q)-2cs(P\vee Q)\leq cs(P\wedge Q)-cs(P\vee Q)\leq 2cs(P\wedge Q)-cs(P)-cs(Q)
\end{equation*}
for all $P,Q\in\mathcal P^N$.
\end{proof}

Denote by $(\Im ,\sqcap ,\sqcup)$ the lattice whose elements are the $\mathcal B_n$ fields of subsets $\mathcal F_P$ generated by partitions $P\in\mathcal P^N$, ordered by
inclusion $\supseteq$. The meet and join are, respectively, $\mathcal F_P\sqcap\mathcal F_Q=\mathcal F_{P\vee Q}$ and $\mathcal F_P\sqcup\mathcal F_Q=\mathcal F_{P\wedge Q}$.
The set of atoms is the collection $\{\mathcal F_{\{A,A^c\}}:\emptyset\subset A\subset N\}$ of minimal fields; that is to say,
$\mathcal F_P=\underset{\{A,A^c\}\geqslant P}{\sqcup}\mathcal F_{\{A,A^c\}}$ for all $\mathcal F_P\in\Im$. Therefore, $\delta_{cs}(P,Q)$ may also be regarded as an analog of
the traditional Hamming distance between subsets:
\begin{eqnarray*}
\delta_{cs}(P,Q)&=&|\{\{A,A^c\}:\mathcal F_{\{A,A^c\}}\subseteq\mathcal F_P\}|+|\{\{A,A^c\}:\mathcal F_{\{A,A^c\}}\subseteq\mathcal F_Q\}|+\\
&-&2|\{\{A,A^c\}:\mathcal F_{\{A,A^c\}}\subseteq(\mathcal F_Q\cap\mathcal F_P)\}|\text .
\end{eqnarray*}
In words, this is the number of minimal fields $\mathcal F_{\{A,A^c\}}$ included in either $\mathcal F_P$ or else in $\mathcal F_Q$, but not in both.

\section{Appendix: Euclidean distance between fuzzy partitions}
The leading idea of this section is to propose a measure of the distance between fuzzy partitions, like in \cite{Brouwer2009}. Together with theoretical worthiness, from an applicative
perspective this distance is useful for comparing alternative results of objective function-based fuzzy clustering algorithms (such as the fuzzy
C-means, see \cite{FuzzyCmeansBook2008,Valente+07} for a comprehensive treatment). More precisely, these algorithms usually rely on local search methods, and their output takes the
form of a membership matrix, where rows and columns are indexed by data and clusters, respectively. For a given data set, the chosen algorithm typically outputs different membership
matrices depending on alternative initial candidate solutions and/or parametrizations, and these varying outputs are commonly ranked through a validity index
(see \cite{WangZhang2007} for a recent overview). A key input is the desired number of clusters, which is not chosen autonomously through optimization, but is conversely maintained
fixed over the search. Conceiving several runs for each reasonable number of clusters, a common situation is thus one where alternative outputs score best on the chosen validity index.
Then, the proposed distance measure allows to compare these outputs, each with a different number of clusters and with highest validity score for that number.

Fuzzy clusterings are collections $A_1,\ldots ,A_m\subseteq N$ of subsets of $N$ endowed with $n$ membership distributions $x_{il},1\leq i\leq n,1\leq l\leq m$, where $x_{il}\in[0,1]$
quantifies the membership of $i\in N$ in $A_l,1\leq l\leq m$. A fuzzy clustering thus is a $m$-collection of fuzzy subsets $(x_{1l},\ldots ,x_{nl})\in[0,1]^n,1\leq l\leq m$ of $N$
\cite{FuzzySubset}, and $m\in\{1,\ldots ,2^n-1\}$ since \textit{every} non-empty subset $A_l\neq\emptyset$ may have an associated fuzzy subset $(x_{1l},\ldots ,x_{nl})\in[0,1]^n$.
Membership matrices $\textbf x\in[0,1]^{n\times m}$ satisfy $\sum_{1\leq l\leq m}x_{il}=1$ for all $i\in N$. The traditional Euclidean distance $d(x,x')$ between
$x=(x_1,\ldots ,x_n),x'=(x'_1,\ldots ,x'_n)\in[0,1]^n$ simply is $d(x,x')=\sqrt{\sum_{1\leq i\leq n}\left(x_i-x'_i\right)^2}$, i.e. the $\ell_2$ norm in $\mathbb R^n$. For measuring the
same distance $d(\textbf x,\textbf x')$ between $\textbf x\in[0,1]^{n\times m}$ and $\textbf x'\in[0,1]^{n\times m'}$ it must be $m=m'$, in which case
$d(\textbf x,\textbf x')=\sqrt{\sum_{1\leq i\leq n,1\leq l\leq m}(x_{il}-x'_{il})^2}$. Yet, as already observed, very likely there are fuzzy clusterings with high scores
in terms of the chosen validity index such that $m\neq m'$. In this view, the proposed method is dimension-free, i.e. regardless of whether $m=m'$ or $m\neq m'$.

When considering that the $n$ singletons $\{i\}$, $i\in N$
are the atoms of Boolean lattice $(2^N,\cap,\cup)$, a fuzzy subset is readily seen to consist, in fact, of $n$ memberships $x_i\in[0,1],1\leq i\leq n$ indexed by the $n$ atoms. In
this view, from a combinatorial
perspective fuzzy elements of atomic lattices may be defined to be collections of [0,1]-memberships, one for each atom. Insofar as lattice theory is concerned, fuzzy partitions may thus be
regarded as points in the 0/1-polytope $\mathbb P_N$ introduced in Section 5, with variables $y_{[ij]_k}$
indexed by atoms $[ij]_k\in\mathcal P^N_{\mathcal A},1\leq k\leq\binom{n}{2}$.

\begin{definition}
A fuzzy partition is any $y=\left(y_{[ij]_1},\ldots ,y_{[ij]_{\binom{n}{2}}}\right)\in\mathbb P_N$,
and $y=I_P\in ex(\mathbb P_N)$ is in fact non-fuzzy (or hard), while $y\in\mathbb P_N\backslash ex(\mathbb P_N)$ is properly fuzzy.
\end{definition}

A fuzzy partition thus is a point in the polytope $\mathbb P_N\subset[0,1]^{\binom{n}{2}}$ included in the $\binom{n}{2}$-cube, with axes indexed by atoms $[ij]\in\mathcal P^N_{\mathcal A}$, and a
non-fuzzy partition $P\in\mathcal P^N$ corresponds to a vertex of $\mathbb P_N$ identified by indicator function $I_P$. 

Denote by $\mathbb M_N=\underset{1\leq m\leq 2^n-1}{\cup}[0,1]^{n\times m}$ the set of all membership matrices. For $\textbf x=[x_{il}]_{1\leq i\leq n,1\leq l\leq m}\in\mathbb M_N$,
let $A_1,\ldots ,A_m$ be the associated subsets, i.e. $x_{il}$ is the membership of $i$ in $A_l$, while $\sum_{1\leq l\leq m}x_{il}=1$ for all $i$. Thus, for instance, if $m=1$, then
$x_{i1}=1$ for $1\leq i\leq n$.

\begin{proposition}
Mapping $\eta:\mathbb M_N\rightarrow[0,1]^{\binom{n}{2}}$ defined by
\begin{equation}
\eta_{[ij]}(\textbf x)=\underset{i,j\in A_l}{\sum_{1\leq l\leq m}}(x_{il}\cdot x_{jl})\text{ for all }[ij]\in\mathcal P^N_{\mathcal A}
\end{equation}
satisfies: (i) if $x_{il}\in\{0,1\}$ for all $i,l$, then $\eta(\textbf x)\in ex(\mathbb P_N)$, and (ii) $\eta(\textbf x)\in\mathbb P_N$.
\end{proposition}

\begin{proof}
Firstly, $\sum_{1\leq l\leq m}x_{il}=1$ for all $i\in N$ entails that the summation yields a positive quantity never exceeding 1, that is $\eta(\textbf x)\in[0,1]^{\binom{n}{2}}$.
Concerning (i), if $x_{il}\in\{0,1\},1\leq i\leq n,1\leq l\leq m$, then $\textbf x$ corresponds to a non-fuzzy partition $P\in\mathcal P^N$, i.e.
$x_{il}=\left\{\begin{array}{c}1\text{ if }i\in A_l\text ,\\0\text{ if }i\notin A_l\text , \end{array}\right .1\leq i\leq n,1\leq l\leq m$.
The $m$ columns $(x_{1l},\ldots ,x_{nl})^T,1\leq l\leq m$ of \textbf x are thus given by the $m$ characteristic functions
$\chi_{A_l},1\leq l\leq m$ of subsets $A_1,\ldots ,A_m$ (see above), with $P=\{A_1,\ldots ,A_m\}$ for some partition $P\in\mathcal P^N$. Hence $\eta(\textbf x)=I_P$, in that
\begin{equation*}
\eta_{[ij]}(\textbf x)=I_P([ij])=\left\{\begin{array}{c}1\text{ if }\{i,j\}\subseteq A_l\text{ for some }l\in\{1,\ldots ,m\}\text ,\\
0\text{ if }\{i,j\}\not\subseteq A_l\text{ for all }l\in\{1,\ldots ,m\}\text , \end{array}\right .
\end{equation*}
for all atoms $[ij]\in\mathcal P^N_{\mathcal A}$.
Finally, coming to (ii), observe that $\eta(\textbf x)$ obtains as a suitable convex combination of vertices $I_{P_1},\ldots ,I_{P_h}\in ex(\mathbb P_N)$ of the polytope. That is to say,
$\eta(\textbf x)=\alpha_{P_1}I_{P_1}+\cdots+\alpha_{P_h}I_{P_h}$ with $\alpha_{P_1},\ldots ,\alpha_{P_h}>0$ and $\sum_{1\leq h'\leq h}\alpha_{P_{h'}}=1$.
These partitions $P_{h'}$ and coefficients $\alpha_{P_{h'}},1\leq h'\leq h$ are determined through a
fairly simple recursive procedure. Starting from the top partition $P_1=P^{\top}$, with coefficient $\alpha_{P^{\top}}=\underset{[ij]\in\mathcal P^N_{\mathcal A}}{\min}\eta_{[ij]}(\textbf x)$,
let $[ij]^1$ be the atom corresponding to this minimum. Next, atom $[ij]^2$ corresponds to minimum $\underset{[ij]\neq[ij]^1}{\min}\eta_{[ij]}(\textbf x)$ and
$P_2<P^{\top}$ is a coarsest partition satisfying $[ij]^1\not\leqslant P_2\geqslant[ij]^2$, while coefficient $\alpha_{P_2}=\eta_{[ij]^2}(\textbf x)-\eta_{[ij]^1}(\textbf x)$ obtains
incrementally. At the generic $h'$-th step, atom $[ij]^{h'}$ corresponds to minimum $\underset{[ij]\neq[ij]^1,\ldots ,[ij]^{h'-1}}{\min}\eta_{[ij]}(\textbf x)$, while the selected partition
$P_{h'}$ is a coarsest one satisfying $[ij]^1,\ldots ,[ij]^{h'-1}\not\leqslant P_h\geqslant[ij]^{h'}$ and the coefficient $\alpha_{P_{h'}}$ is given by
$\alpha_{P_{h'}}=\eta_{[ij]^{h'}}(\textbf x)-\eta_{[ij]^{h'-1}}(\textbf x)$. These steps continue through partitions that are either finer or else incomparable with respect to the previous
ones, while reaching the atoms themselves and, if necessary, the bottom $P_{\bot}$ too.
\end{proof}

\begin{example}
For $N=\{1,2,3,4\}$, consider the collections

$\{A_1,A_2,A_3\}=\{\{1,2,3\},\{1,4\},\{2,3,4\}\}$ and

$\{A'_1,A'_2,A'_3,A'_4\}=\{\{1,2\},\{2,3\},\{3,4\},\{1,2,3,4\}\}$ of subsets, with membership matrices $\textbf x\in[0,1]^{4\times 3}$ and $\textbf x'\in[0,1]^{4\times 4}$ given by:

$x_{11}=0.7$, $x_{12}=0.3$, $x_{13}=0$ and

$x_{21}=0.4$, $x_{22}=0$, $x_{23}=0.6$ and

$x_{31}=0.2$, $x_{32}=0$, $x_{33}=0.8$ and

$x_{41}=0$, $x_{42}=0.5$, $x_{43}=0.5$ for the former, while

$x'_{11}=0.4$, $x'_{12}=0=x'_{13}$, $x'_{14}=0.6$ and

$x'_{21}=0.2$, $x'_{22}=0.3$, $x'_{23}=0$, $x'_{24}=0.5$ and

$x'_{31}=0$, $x'_{32}=0.3$, $x'_{33}=0.4$, $x'_{34}=0.3$ and

$x'_{41}=0=x'_{42}$, $x'_{43}=0.8$, $x'_{44}=0.2$ for the latter. Let $\eta_{[ij]}(\textbf x)=y_{[ij]}=y_{ij}$ for notational convenience. Expression (10) yields:

$y_{12}=0.7\cdot 0.4=0.28$,

$y_{13}=0.7\cdot 0.2=0.14$,

$y_{14}=0.3\cdot 0.5=0.15$,

$y_{23}=0.4\cdot 0.2+0.6\cdot0.8=0.08+0.48=0.56$,

$y_{24}=0.6\cdot 0.5=0.3$,

$y_{34}=0.8\cdot0.5=0.4$ for the former collection, and

$y'_{12}=0.4\cdot 0.2+0.6\cdot 0.5=0.08+0.3=0.38$,

$y'_{13}=0.6\cdot 0.3=0.18$,

$y'_{14}=0.6\cdot 0.2=0.12$,

$y'_{23}=0.3\cdot 0.3+0.5\cdot0.3=0.09+0.15=0.24$,

$y'_{24}=0.5\cdot 0.2=0.1$,

$y'_{34}=0.4\cdot0.8+0.3\cdot0.2=0.32+0.06=0.38$ for the latter. Concerning the convex combinations corresponding to $\eta(\textbf x)=y$
and $\eta(\textbf x')=y'$, for the former collection $\{\{1,2,3\},\{1,4\},\{2,3,4\}\}$, since $y_{13}\leq y_{ij},1\leq i<j\leq 4$, firstly $P_1=P^{\top}$ and
$\alpha_{P_1}=\alpha_{P^{\top}}=0.14=y_{13}$, thus partitions $P_2,P_3,\ldots$ coming next satisfy $P_2,P_3,\ldots\not\geqslant[13]$.
As $y_{14}=0.15<y_{ij}$ for $[ij]\neq[13],[14]$, a coarsest $P\geqslant[14]$ is $P_2=124|3$. Hence
$\alpha_{P_2}=\alpha_{124|3}=0.15-0.14=0.01$
and therefore $P\not\geqslant[13],[14]$ for all subsequest partitions $P$.
The new minimum is $y_{12}=0.28$, and the above constraints yield  $P_3=12|34$ as the coarsest available partition, with
$\alpha_{12|34}=0.28-0.15=0.13$. After updating, $y_{34}=0.4$ is the novel minimum, with $P_4=1|234$ and
$\alpha_{1|234}=0.4-0.13-0.14=0.13$.
The last partitions $P_5,P_6$ are atoms themselves, namely $P_5=[24]$ and $P_6=[23]$, with associated coefficients
$\alpha_{1|24|3}=0.3-0.13-0.01-0.14=0.02$ as well as
$\alpha_{1|23|4}=0.56-0.13-0.14=0.29$.
Since the sum of these six coefficients yields $0.72$, the bottom partition finally has coefficient
$\alpha_{P_{\bot}}=1-0.72=0.28$.
Thus the sought convex combination of indicator functions or vertices $I_P\in ex(\mathbb P_N)$ is 
\begin{eqnarray*}
y&=&0.14\cdot I_{1234}+0.01\cdot I_{124|3}+0.13\cdot I_{12|34}++0.13\cdot I_{1|234}+0.02\cdot I_{1|24|3}+\\
&+&0.29\cdot I_{1|23|4}+0.28\cdot I_{1|2|3|4}\text .
\end{eqnarray*}
A generic point in polytope $\mathbb P_N$ generally admits alternative (equivalent)
convex combinations of vertices. For instance, $y$ also admits
\begin{eqnarray*}
y&=&0.15\cdot I_{14|2|3}+0.14\cdot I_{123|4}+0.14\cdot I_{12|3|4}++0.3\cdot I_{1|234}+0.1\cdot I_{1|2|34}+\\
&+&0.12\cdot I_{1|23|4}+0.05\cdot I_{1|2|3|4}\text .
\end{eqnarray*}
Coming to the second collection $\{\{1,2\},\{2,3\},\{3,4\},\{1,2,3,4\}\}$ of subsets, the first coefficient is $\alpha'_{P^{\top}}=y'_{24}=0.1$
since $y'_{24}<y'_{ij},1\leq i<j\leq 4$. Next, rather straightforwardly,

$\alpha'_{[12]}=\alpha'_{12|3|4}=y'_{12}-y'_{24}=0.38-0.1=0.28$,

$\alpha'_{[13]}=\alpha'_{13|2|4}=y'_{13}-y'_{24}=0.18-0.1=0.08$,

$\alpha'_{[14]}=\alpha'_{14|2|3}=y'_{14}-y'_{24}=0.12-0.1=0.02$,

$\alpha'_{[23]}=\alpha'_{1|23|4}=y'_{23}-y'_{24}=0.24-0.1=0.14$,

$\alpha'_{[34]}=\alpha'_{1|2|34}=y'_{34}-y'_{24}=0.38-0.1=0.28$. These six coefficients add up to $0.9$, hence the bottom partition has coefficient
$\alpha'_{P_{\bot}}=1-0.9=0.1$. A sought convex combination thus is
\begin{eqnarray*}
y'&=&0.1\cdot I_{1234}+0.02\cdot I_{14|2|3}+0.08\cdot I_{13|2|4}+0.14\cdot I_{1|23|4}+0.28\cdot I_{12|3|4}+\\
&+&0.28\cdot I_{1|2|34}+0.1\cdot I_{1|2|3|4}\text .
\end{eqnarray*}
\end{example}

The Euclidean distances between fuzzy partitions $y,y'\in\mathbb P_N$ given by the $\ell_1$ and $\ell_2$ norms, denoted by $d_1(y,y')$ and $d_2(y,y')$ respectively,
are the usual distances between points in a Euclidean vector space (i.e. $\mathbb R^{\binom{n}{2}}$), namely
\begin{equation*}
d_1(y,y')=\sum_{[ij]\in\mathcal P^N_{\mathcal A}}abs\left(y_{[ij]}-y'_{[ij]}\right)\text{ and }d_2(y,y')=\sqrt{\sum_{[ij]\in\mathcal P^N_{\mathcal A}}\left(y_{[ij]}-y'_{[ij]}\right)^2}\text ,
\end{equation*}
where $abs(\alpha-\beta)=\max\{\alpha,\beta\}-\min\{\alpha,\beta\}$ is the absolute value. Both are well-known metrics (see above). In particular, triangle inequality
may be considered in conjunction with the order relation and the meet of fuzzy partitions.

\subsection{Order, meet and join}
The order relation $\geqslant$, the meet $\wedge$ and the join $\vee$ for partitions $P,Q\in\mathcal P^N$ may be extended from vertices $I_P,I_Q$ of polytope $\mathbb P_N$
to the whole of this latter. Specifically, $P\geqslant Q\Leftrightarrow I_P([ij])\geq I_Q([ij])\text{ for all }[ij]\in\mathcal P^N_{\mathcal A}$. In the same way,
for any two fuzzy partitions $y,y'\in\mathbb P_N$,
\begin{equation*}
y\geqslant y'\Leftrightarrow y_{[ij]}\geq y'_{[ij]}\text{ for all }[ij]\in\mathcal P^N_{\mathcal A}\text .
\end{equation*}
For the discrete setting provided by vertices of the polytope, the following condition has been already considered in terms of ``vertical collinearity'' \cite{Meila2007} or
``interval betweenness'' (for valuations of distributive lattices) \cite{Monjardet1981}.

\begin{proposition}
For any fuzzy partitions $y,z,y'\in\mathbb P_N$, if $y\geqslant z\geqslant y'$, then $d_1(y,z),d_1(z,y')$ and $d_1(y,y')$ satisfy triangle inequality with equality, that is
\begin{equation*}
d_1(y,z)+d_1(z,y')=d_1(y,y')\text .
\end{equation*}
\end{proposition}

\begin{proof}
If $y\geqslant z\geqslant y'$, then for all atoms $[ij]\in\mathcal P^N_{\mathcal A}$
\begin{eqnarray*}
abs\left(y_{[ij]}-z_{[ij]}\right)&=&y_{[ij]}-z_{[ij]}\text ,\\
abs\left(z_{[ij]}-y'_{[ij]}\right)&=&z_{[ij]}-y'_{[ij]}\text ,\\
abs\left(y_{[ij]}-y'_{[ij]}\right)&=&y_{[ij]}-y'_{[ij]}\text ,
\end{eqnarray*}
and of course
\begin{equation*}
y_{[ij]}-z_{[ij]}+z_{[ij]}-y'_{[ij]}=y_{[ij]}-y'_{[ij]}\text ,
\end{equation*}
hence $d_1(y,z)+d_1(z,y')=d_1(y,y')$.
\end{proof}

The same does not hold for $d_2$, which conversely satisfies triangle inequality with equality if and only if $z$ lies on the line segment between $y$ and $y'$.

Turning attention to the meet $y\wedge y'$ of fuzzy partitions $y,y'\in\mathbb P_N$, firstly consider that for the characteristic functions $\chi_A,\chi_B,\chi_{A\cap B},A,B\in 2^N$
of subsets the meet or intersection is given by $\chi_{A\cap B}(i)=\chi_A(i)\chi_B(i)$ for all $i\in N$, i.e. by the product. Analogously, the meet of partitions
$P,Q,P\wedge Q\in\mathcal P^N$, with indicator functions $I_P,I_Q,I_{P\wedge Q}\in ex(\mathbb P_N)$, is given by
\begin{equation*}
I_{P\wedge Q}([ij])=I_P([ij])I_{Q}([ij])\text{ for all }[ij]\in\mathcal P^N_{\mathcal A}\text ,
\end{equation*}
as $P\wedge Q=\underset{P\geqslant[ij]\leqslant Q}{\vee}[ij]$. Then, the meet $y\wedge y'$ of fuzzy partitions $y,y'\in\mathbb P_N$ also obtains through the product:
$(y\wedge y')_{[ij]}=y_{[ij]}y'_{[ij]}\text{ for all }[ij]\in\mathcal P^N_{\mathcal A}$.

\begin{proposition}
For all $y,y'\in\mathbb P_N$,
\begin{equation*}
d_1(y,y\wedge y')+d_1(y\wedge y',y')-d_1(y,y')=2\sum_{[ij]\in\mathcal P^N_{\mathcal A}}\left(\min\left\{y_{[ij]},y'_{[ij]}\right\}-y_{[ij]}y'_{[ij]}\right)\text .
\end{equation*}
\end{proposition}

\begin{proof}
Firstly note that $y_{[ij]},y'_{[ij]}\in[0,1]$ entails $y_{[ij]}\geq y_{[ij]}y'_{[ij]}\leq y'_{[ij]}$.

Now, $d_1(y,y\wedge y')+d_1(y\wedge y',y')-d_1(y,y')=$
\begin{eqnarray*}
&=&\sum_{[ij]\in\mathcal P^N_{\mathcal A}}\Big[y_{[ij]}-y_{[ij]}y'_{[ij]}+y'_{[ij]}-y_{[ij]}y'_{[ij]}-abs\left(y_{[ij]}-y'_{[ij]}\right)\Big]=\\
&=&\sum_{[ij]\in\mathcal P^N_{\mathcal A}}\Big[\max\left\{y_{[ij]},y'_{[ij]}\right\}+\min\left\{y_{[ij]},y'_{[ij]}\right\}-2y_{[ij]}\cdot y'_{[ij]}+\\
&-&\left(\max\left\{y_{[ij]},y'_{[ij]}\right\}-\min\left\{y_{[ij]},y'_{[ij]}\right\}\right)\Big]=
\end{eqnarray*}
$=2\sum_{[ij]\in\mathcal P^N_{\mathcal A}}\left(\min\left\{y_{[ij]},y'_{[ij]}\right\}-y_{[ij]}y'_{[ij]}\right)$ as wanted.
\end{proof}

A similar expression may be provided for the squared Euclidean distance or $\ell_2^2$ norm $d_2^2\left(y,y\wedge y'\right)+d_2^2\left(y\wedge y',y'\right)-d_2^2\left(y,y'\right)$.

\begin{proposition}
$d_2^2\left(y,y\wedge y'\right)+d_2^2\left(y\wedge y',y'\right)-d_2^2\left(y,y'\right)=$
\begin{equation*}
=2\sum_{[ij]\in\mathcal P^N_{\mathcal A}}\left[y_{[ij]}y'_{[ij]}\left(1-y_{[ij]}-y'_{[ij]}+y_{[ij]}y'_{[ij]}\right)\right]\text{ for all }y,y'\in\mathbb P_N\text .
\end{equation*}
\end{proposition}

\begin{proof}
By direct substitution: $d_2^2\left(y,y\wedge y'\right)+d_2^2\left(y\wedge y'\right)-d_2^2\left(y,y'\right)=$
\begin{eqnarray*}
&=&\sum_{[ij]\in\mathcal P^N_{\mathcal A}}\Big[\left(y_{[ij]}-y_{[ij]}y'_{[ij]}\right)^2+\left(y_{[ij]}y'_{[ij]}-y'_{[ij]}\right)^2-\left(y_{[ij]}-y'_{[ij]}\right)^2\Big]=\\
&=&\sum_{[ij]\in\mathcal P^N_{\mathcal A}}\Big[2y^2_{[ij]}y'^2_{[ij]}-2y_{[ij]}y'_{[ij]}\left(y_{[ij]}+y'_{[ij]}\right)+2y_{[ij]}y'_{[ij]}\Big]=
\end{eqnarray*}
$=2\sum_{[ij]\in\mathcal P^N_{\mathcal A}}\left[y_{[ij]}y'_{[ij]}\left(1-y_{[ij]}-y'_{[ij]}+y_{[ij]}y'_{[ij]}\right)\right]$ as wanted.
\end{proof}

The join $\vee$ of two (fuzzy) partitions leads to a more complex setting, because it brings about the closure yielding the partition lattice as the
polygon matroid defined on the edges of the complete graph $K_N=(N,N_2)$ (see Sections 1, 2). As already observed, for $P,Q\in\mathcal P^N$, the meet
$P\wedge Q=\underset{P\geqslant[ij]\leqslant Q}{[ij]}$ is coarser than all and only those atoms
$[ij]\in\mathcal P^N_{\mathcal A}$ finer than both $P$ and $Q$. Thus, the meet of partitions is basically the analog of the intersection of subsets $A,B\in 2^N$, and indeed
in the same way obtains through the pair-wise product of indicator functions $I_P,I_Q$ (see above). Conversely, when regarded as a pair-wise operation between indicator
functions $I_P,I_Q$, the join is very different from the union of subsets. In particular, recall that for $A,B\in 2^N$, with characteristic functions
$\chi_A,\chi_B\in\{0,1\}^n$, the union $A\cup B$ obtains as follows: $\chi_{A\cup B}(i)=\max\{\chi_A(i),\chi_B(i)\}\text{ for all }i\in N$.
In words, Boolean vector $\chi_{A\cup B}\in\{0,1\}^n$ has entry 1 where $\chi_A$ and/or $\chi_B$ have entry 1. The same does not apply to partitions $P,Q\in\mathcal P^N$,
as indicator function $I_{P\vee Q}\in\{0,1\}^{\binom{n}{2}}$ may have entry 1 even where both $I_P$ and $I_Q$ have entry 0. As before,
this can be observed already in the simple case where $N=\{1,2,3\}$. To this end, arrange the entries of $I_P,P=P_{\bot},[12],[13],[23],P^{\top}$ by
$I_P=\left(\begin{array}{c}I_P([12])\\ I_P([13])\\ I_P([23])\end{array}\right)$, hence $I_{[12]}=\left(\begin{array}{c}1\\ 0\\ 0\end{array}\right)$ and
$I_{[23]}=\left(\begin{array}{c}0\\ 0\\ 1\end{array}\right)$. Then,
$\max\left\{I_{[12]},I_{[23]}\right\}=\left(\begin{array}{c}1\\ 0\\ 1\end{array}\right)$, but $I_{[12]\vee[23]}=\left(\begin{array}{c}1\\ 1\\ 1\end{array}\right)$,
i.e. $I_{[12]\vee[23]}\geqslant\max\left\{I_{[12]},I_{[23]}\right\}$.

\begin{definition}
In terms of indicator functions $I_P,I_Q,I_{P\vee Q}\in\{0,1\}^{\binom{n}{2}}$, for all atoms $[ij]\in\mathcal P^N_{\mathcal A}$
the join $P\vee Q$ of partitions $P,Q\in\mathcal P^N$ is
\begin{equation*}
I_{P\vee Q}([ij])=\max\left\{I_P([ij]),I_Q([ij]),\underset{i'\in N\backslash\{i,j\}}{\max}I_P([ii'])I_Q([ji'])\right\}\text .
\end{equation*}
\end{definition}

In the same way, the join $\left(y\vee y'\right)_{[ij]}$ of fuzzy partitions $y,y'\in\mathbb P_N$ is given by
$\left(y\vee y'\right)_{[ij]}=\max\left\{y_{[ij]},y'_{[ij]},\underset{i'\in N\backslash\{i,j\}}{\max}y_{[ii']}y'_{[ji']}\right\}$.

\section{Appendix: the consensus partition problem}
Hamming distance between partitions HD was considered for the first time in the mid '60s \cite{Renier1965} in terms of the \textit{consensus (or central) partition} problem, which
is important in many applicative scenarios concerned with statistical classification. From a combinatorial optimization perspective, the problem has generic instance consisting of
a $m$-collection $P_1,\ldots ,P_m\in\mathcal P^N$, $m\geq 2$, and is characterized by firstly selecting a measure of the distance between any two partitions, i.e. a metric
$\delta:\mathcal P^N\times\mathcal P^N\rightarrow\mathbb R_+$. Given this, the objective is to find a partition $\hat P$ minimizing the sum of its distances from the $m$
partitions. That is to say, any $\hat P$ satisfying $\sum_{1\leq k\leq m}\delta(\hat P,P_m)\leq\sum_{1\leq k\leq m}\delta(Q,P_k)$ for all $Q\in\mathcal P^N$ is a consensus
partition. For generic $\delta$, finding a solution $\hat P$ is tipically hard. In particular, if $\delta=MMD$, then each distance $\delta(Q,P_k),1\leq k\leq m$ for any
$Q\in\mathcal P^N$ is computable in $\mathcal O(n^3)$ time \cite[p. 236]{KorteVygen2002}, whereas if $\delta=HD$, then in view of expression (6) above (see Section 3) each
distance $\delta(Q,P_k)$ is computable more rapidly through scalar products. In any case, independently from the chosen metric $\delta$, the main issue is that the size
$\mathcal B_n=|\mathcal P^N|$ of the search space $\mathcal P^N$ makes all approaches relying on direct enumeration simply unviable, at least for relevant values of $n$. The
problem is thus commonly interpreted in terms of heuristics \cite{CeleuxEtAl1989,CentralPartition}, and if $m$ is large and/or $P_1,\ldots ,P_m$ are very far from each other, then
figuring out where to concentrate the search is the fundametal issue.

Although the consensus problem is generally harsh, especially in terms of the required exploration of $\mathcal P^N$, still the analysis conducted thus far identifies conditions
where exact solutions are easy to find. In fact, if the chosen metric is a minimum-$f$-weight partition distance, i.e. $\delta=\delta_f$ with $f\in\mathbb F$, and weighting
function $f$ is either supermodular or else submodular (but not both, see below), then either the meet $\hat P=P_1\wedge\cdots\wedge P_m$ or else the join
$\hat P=P_1\vee\cdots\vee P_m$ of instance elements are consensus partitions. Specifically, the former case applies to Hamming distance or size-based $\delta_s=HD$ and to logical
entropy-based $\delta_h$, while the latter applies to rank-based $\delta_r$ and to co-size-based $\delta_{cs}$. Hence, the computational burden reduces solely to assessing the
$m$ distances between instance elements and their meet (or else their join), with no search need.

\begin{proposition}
If distances between partitions are measured by HD, then the meet of all instance elements achieves consensus, i.e.
\begin{equation*}
\sum_{1\leq k\leq m}HD(P_1\wedge\cdots\wedge P_m,P_k)\leq\sum_{1\leq k\leq m}HD(Q,P_k)
\end{equation*}
for all $Q\in\mathcal P^N$ and all instances $\mathcal I=\{P_1,\ldots ,P_m\}\subseteq\mathcal P^N$.
\end{proposition}

\begin{proof}
Firstly note that for $m=2$ this consensus condition is in fact a restatement of horizontal collinearity and triangle inequality (see Propositions 1 and 2). Hence,
in order to use induction, assume that the condition holds for some $m\geq 2$, and denote by $\hat P$ the solution or consensus partition of a $m+1$-instance
$P_1\ldots,P_m,P_{m+1}$. By assumption, $P_1\wedge\cdots\wedge P_m$ is a solution of instance $P_1,\ldots ,P_m$, thus novel solution $\hat P$ minimizes the sum of its
distances from the previous solution $P_1\wedge\cdots\wedge P_m$ and from the novel instance element $P_{m+1}$, i.e.
\begin{equation*}
HD(P_1\wedge\cdots\wedge P_m,\hat P)+HD(\hat P,P_{m+1})\leq HD(P_1\wedge\cdots\wedge P_m,Q)+HD(Q,P_{m+1})
\end{equation*}
for all $Q\in\mathcal P^N$. Then, horizontal collinearity and triangle inequality entail
\begin{equation*}
HD(P_1\wedge\cdots\wedge P_m,\hat P)+HD(\hat P,P_{m+1})\geq HD(P_1\wedge\cdots\wedge P_m,P_{m+1})\text ,
\end{equation*}
with equality if $\hat P=P_1\wedge\cdots\wedge P_m\wedge P_{m+1}$.
\end{proof}  

Concerning the value taken by the sum $\sum_{1\leq k\leq m}HD(P_k,P_1\wedge\cdots\wedge P_m)$ of distances between instance elements and the consensus partition,
observe that for all $Q\in\mathcal P^N$ and all $\mathcal I=\{P_1,\ldots ,P_m\}$
\begin{equation*}
\sum_{1\leq k\leq m}HD(Q,P_k)=\sum_{1\leq k<k'\leq m}\frac{HD(P_k,Q)+HD(Q,P_{k'})}{m-1}\text .
\end{equation*}
By triangle inequality,
\begin{equation*}
\sum_{1\leq k<k'\leq m}\frac{HD(P_k,Q)+HD(Q,P_{k'})}{m-1}\geq\sum_{1\leq k<k'\leq m}\frac{HD(P_k,P_{k'})}{m-1}\text ,
\end{equation*}
with equality if $Q=P_k\wedge P_{k'}$ for all $1\leq k<k'\leq m$, which is not possible unless $m=2$. Now consider partition function
$\mathcal{D_I}:\mathcal P^N\rightarrow\mathbb R_+$ defined by
\begin{eqnarray*}
\mathcal{D_I}(Q)&=&
\frac{1}{m-1}\sum_{1\leq k<k'\leq m}[HD(P_k,Q)+HD(Q,P_{k'})-HD(P_k,P_{k'})]\\
&=&\frac{2}{m-1}\sum_{1\leq k<k'\leq m}[s(Q)-s(P_k\wedge Q)-s(P_{k'}\wedge Q)+s(P_k\wedge P_{k'})]\text ,
\end{eqnarray*}
where $\mathcal I=\{P_1,\ldots ,P_m\}$ denotes the given instance. Function $\mathcal{D_I}$ attains its minimum at consensus partition
$\hat P_{\mathcal I}:=P_1\wedge\cdots\wedge P_m$, where
\begin{equation*}
\mathcal{D_I}(\hat P_{\mathcal I})=\frac{2}{m-1}\sum_{1\leq k<k'\leq m}[s(P_k\wedge P_{k'})-s(\hat P_{\mathcal I})]
\end{equation*}
as $HD(P_k,\hat P_{\mathcal I})+HD(\hat P_{\mathcal I},P_{k'})=HD(P_k,P_{k'})+2[s(P_k\wedge P_{k'})-s(\hat P_{\mathcal I})]$ for all $1\leq k<k'\leq m$.

Exactly the same argument applies to logical entropy-based $\delta_h$, entailing that
$\sum_{P\in\mathcal I}\delta_h(P,\hat P_{\mathcal I})\leq\sum_{P\in\mathcal I}\delta_h(P,Q)$ for all $Q\in\mathcal P^N$ and all instances $\mathcal I$.

For rank-based $\delta_r$ and co-size-based $\delta_{cs}$ distances, horizontal collinearity holds in terms of the join (rather than in terms of the meet of any
$P,Q\in\mathcal P^N$, see above), meaning that $\delta\in\{\delta_r,\delta_{cs}\}$ yields
\begin{equation*}
\delta(P,P')+\delta(P',Q)\leq\delta(P,Q)\text{ for all }P,P',Q\in\mathcal P^N\text ,
\end{equation*}
with equality if $P'=P\vee Q$. Thus the join (rather than the meet) of instance elements achieves consesus, i.e.
$\sum_{P\in\mathcal I}\delta(P,\vee_{P\in\mathcal I}P)\leq\sum_{P\in\mathcal I}\delta(P,Q)$ for all $Q\in\mathcal P^N$ and all instances $\mathcal I$, while analog
results apply, \textit{mutatis mutandis}, to partition function $\mathcal{D_I}$.

The setting developed thus far also enables to frame the consensus partition problem in a novel manner, which in turn widens the spectrum of conceivable fuzzy models
for partitions. In order to briefly outline such new possibilities, firstly recall that a fuzzy subset of $N$ is a function $q:N\rightarrow[0,1]$ or,
from an equivalent geometric perspective, a point $q=(q_1,\ldots ,q_n)\in[0,1]^n$ in the $n$-dimensional unit hypercube, where $q_i=q(i)$, $i\in N$. Accordingly, a fuzzy
partition is commonly intended as a partition $P=\{A_1,\ldots ,A_{|P|}\}$ with associated $|P|$ points $q^A\in[0,1]^n,A\in P$ in the hypercube such that $q_i^A\in(0,1]$
for all $i\in A$ and all $A\in P$. On the other hand, a fuzzy graph with vertex set $N$ may be seen as one whose edge set is a fuzzy subset of $N_2$,
i.e. a function $t:N_2\rightarrow[0,1]$ or, from an equivalent geometric perspective, a point in the $\binom{n}{2}$-dimensional unit hypercube, i.e.
$t=\left(t_{{\{i,j\}}_1},\ldots ,t_{{\{i,j\}}_{\binom{n}{2}}}\right)\in[0,1]^{\binom{n}{2}}$.

By looking at partitions of $N$ as graphs with vertex set $N$ each of whose components is complete, fuzzy partitions can be regarded as fuzzy graphs with complete components.
Along this route, the fuzzy consensus partition $t_{\mathcal I}$ associated with instance $\mathcal I\subseteq\mathcal P^N$ may be defined to be the point in the interior of
the polytope $\mathbb P$ of partitions (see above) corresponding to the center of the convex hull $conv(\{I_P:P\in\mathcal I\})$ given by all convex combinations of the
indicator functions $I_P,P\in\mathcal I$ of instance elements. In this way, the fuzzy consensus partition is a function ranging in the unit interval $[0,1]$ and taking values
on the atoms of $\mathcal P^N$, i.e. $t_{\mathcal I}:\mathcal P^N_{(1)}\rightarrow[0,1]$. In particular,
\begin{equation*}
t_{\mathcal I}([ij])=\frac{1}{|\mathcal I|}\sum_{P\in\mathcal I}I_P([ij])\text{ for all atoms }[ij]\in\mathcal P^N_{(1)}\text .
\end{equation*}

In this framework, the \textit{strong patterns} of instance $\mathcal I$ considered in \cite{CentralPartition} are the blocks of partition $P(t_{\mathcal I})$ obtained through
defuzzification of $t_{\mathcal I}$ as follows:
\begin{equation*}
P(t_{\mathcal I})=\underset{t_{\mathcal I}([ij])=1}{\vee}[ij]\text .
\end{equation*}
In words, $P(t_{\mathcal I})$ obtains as the join of all atoms where the fuzzy consensus partition attains its maximum, i.e. 1.

\section{Conclusion}
This work considers distances between partitions
by focusing on lattice theory and relying on discrete methods. Specifically, it firstly develops from the idea of reproducing the traditional Hamming distance between subsets by
counting unordered pairs of partitioned elements or atoms of the partition lattice. Although counting ordered and/or unordered pairs is not new (see \cite[Section 2.1]{Meila2007}
for a survey), still the Hamming distance between partitions HD is here analyzed from a novel geometric perspective. Special attention is placed on the distance between complements
in comparison with two alternative partition distance measures proposed in recent years, namely MMD and VI. Given its low computational complexity combined with fine measurement
sensitivity, HD may be considered as an alternative to MMD and VI for applications.

Like the cardinality of the symmetric difference between subsets is a count of atoms of a Boolean lattice, in the same way HD relies on the size, which counts the atoms
finer than partitions, but while the cardinality or rank of subsets is a valuation, i.e. both supermodular and submodular, the size of partitions is supermodular, in that valuations
of the partition lattice are constant partition functions \cite{Aigner79}. Also, in view of expression $|A\cup B|-|A\cap B|$ for the Hamming distance between subsets
$A,B$, it may seem reasonable to consider distances between partitions $P,Q$ of the form $f(P\vee Q)-f(P\wedge Q)$ for some symmetric and order preserving/inverting $f$,
i.e. $f\in\mathbb F$. However, such a distance takes the same value $f(P^{\top})-f(P_{\bot})$ whenever $P$ and $Q$ are complementary partitions (see Section 4), and this
should be avoided in view of \cite[Theorem 1]{Stanley1971}.

The geometric approach adopted here enables to analyze further partition distances obtained by replacing the size with alternative partition functions such as entropy, rank,
logical entropy and co-size. In general, any symmetric and order-preserving/inverting partition function $f$ provides a distance between partitions $P,Q$ by considering the four
values $f(P),f(Q),f(P\wedge Q)$ and $f(P\vee Q)$. Specifically, $f$ defines weights on edges of the Hasse diagram (or 0/1-polytope) of partitions such that the so-called minimum-$f$-weight
distance between any $P,Q$ is the weight of a lightest $P-Q$-path. Depending on whether $f$ is supermodular or else submodular and order-preserving or else order-inverting, a
minimum-$f$-weight path between $P$ and $Q$ visits their meet or else their join, and viceversa. These four possibilities are summarized in Table 3, Section 5. In particular, HD is the
minimum-$s$-weight distance $\delta_s$, where partition function $s$ is the size, while VI is the minimum-$e$-weight distance $\delta_e$, where partition function $e$ is the entropy.

Any distance is of course normalized when considered as the ratio to its maximum value $d_{\max}$. On the other hand, it may be relevant to consider such a maximum as a function
$d_{\max}(n)$ of the number $n$ of partitioned elements, with focus on the first-order difference $\mathcal Dd_{\max}(n)=d_{\max}(n+1)-d_{\max}(n)$ and on the second-order one
$\mathcal D^2d_{\max}(n)=d_{\max}(n+2)-2d_{\max}(n+1)+d_{\max}(n)$. For HD both differences are strictly positive: $\mathcal DHD_{\max}(n)=n$ and $\mathcal D^2HD_{\max}(n)=1$, and
these are exactly the same values $\mathcal D|A\Delta B|,\mathcal D^2|A\Delta B|$ as for the traditional Hamming distance $|A\Delta B|$ between subsets $A,B$. For entropy-based
distance VI, the former $\mathcal DVI_{\max}(n)=\log(n+1)-\log(n)$ is positive while the latter $\mathcal D^2VI_{\max}(n)=\log(n+2)-2\log(n+1)+\log(n)$ is negative by concavity of
the $\log$ function. For maximum matching distance $\mathcal DMMD_{\max}(n)=1$ while $\mathcal D^2MMD_{\max}(n)=0$, and the same applies to rank-based minimum-weight distance
$\delta_r$ outlined in Example 2, Section 5. For logical entropy-based minimum-weight distance $\mathcal D\delta_{h-max}(n)=\frac{1}{(n+1)n}$ and
$\mathcal D^2\delta_{h-max}(n)=\frac{-2}{(n+2)(n+1)n}$.

By extending attention from edges and vertices of the 0/1-polytope of partitions to the whole of this latter, the general aproach based on atoms also applies to the fuzzufication
of partitions. In particular, fuzzy clusterings or membership matrices of any dimension are turned into fuzzy subsets of atoms of the partition lattice, and thus distances between
such matrices may be computed through common Euclidean norm in $\mathbb R^{\binom{n}{2}}$.

\bibliographystyle{abbrv}
\bibliography{HammingDistanceLatticeMetrics}

\end{document}